\documentclass[pra,aps,10pt,twocolumn,showkeys]{revtex4-2}

\pdfoutput=1

\usepackage{amsmath}
\usepackage{amsfonts}
\usepackage{amssymb}
\usepackage{graphicx}
\usepackage{mathtools}
\usepackage{braket}
\usepackage[colorlinks=true,linkcolor=blue,citecolor=blue,plainpages=false,pdfpagelabels]
{hyperref}
\usepackage{bm}

\providecommand{\U}[1]{\protect\rule{.1in}{.1in}}

\newtheorem{theorem}{Theorem}

\newtheorem{lemma}{Lemma}

\newtheorem{proposition}{Proposition}
\newtheorem{remark}{Remark}

\newenvironment{proof}[1][Proof]{\noindent\textbf{#1.} }{\ \rule{0.5em}{0.5em}}
\allowdisplaybreaks

\renewcommand{\vec}[1]{\bm{#1}}
\newcommand{\mat}[1]{\bm{#1}}

\newcommand{\tp}{\mathrm{T}}

\begin{document}

\title{Information geometry of bosonic Gaussian thermal states}
\author{Zixin Huang}
\affiliation{School of Mathematical and Physical Sciences, Macquarie University, NSW 2109, Australia}
\affiliation{Centre for Quantum Software and Information, Faculty of Engineering and
Information Technology, University of Technology Sydney, Ultimo, NSW 2007, Australia}
\author{Mark M. Wilde}
\affiliation{School of Electrical and Computer Engineering, Cornell University, Ithaca, New
York 14850, USA}

\begin{abstract}
Bosonic Gaussian thermal states form a fundamental class of states in quantum information science. This paper explores the information geometry of these states, focusing on characterizing the distance between two nearby states and the geometry induced by a parameterization in terms of their mean vectors and Hamiltonian matrices. In particular, for the family of bosonic Gaussian thermal states, we derive expressions for their Fisher--Bures, Kubo--Mori, and $\alpha$-$z$ information matrices with respect to their mean vectors and Hamiltonian matrices. An important application of our formulas consists of fundamental limits on how well one can estimate these parameters. We additionally establish formulas for the derivatives and the symmetric logarithmic derivatives of bosonic Gaussian thermal states. The former could have applications in gradient descent algorithms for quantum machine learning when using bosonic Gaussian thermal states as an ansatz, and the latter in formulating optimal strategies for single parameter estimation of bosonic Gaussian thermal states. Finally, the expressions for the aforementioned information matrices could have additional applications in  natural gradient descent algorithms when using bosonic Gaussian thermal states as an ansatz.

\end{abstract}
\maketitle
\tableofcontents

\section{Introduction}

\subsection{Motivation}

Bosonic Gaussian states are foundational in quantum information science~\cite{Wang2007,Adesso2014,Ser17,RevModPhys.84.621}. Not only do they generalize Gaussian probability densities from probability theory, but they also are routinely generated in quantum optics laboratories for experiments related to quantum communication~\cite{Lee2011}, sensing~\cite{Aasi2013,Abbott2016}, and computation~\cite{Madsen2022}. They are characterized exclusively by their first and second moments, which makes them far more analytically tractable to deal with than non-Gaussian states. As such, we have witnessed much theoretical progress in understanding their features and properties.

Here we are concerned with a foundational question regarding them; namely, how can we characterize the distance between two nearby bosonic Gaussian states and the geometry realized by a parameterization of them? There are various ways that one can address such a question, but here we adopt the perspective of information geometry of quantum states. To be more precise, let us suppose that a general $n$-mode bosonic Gaussian state is parameterized as
\begin{equation}
    \rho(\mu,H) \coloneqq \frac{1}{Z(\mu,H)} \exp\!\left( -\frac{1}{2} (\hat{x} - \mu)^T H (\hat{x} - \mu)\right),
    \label{eq:intro-bosonic-Gaussian-thermal-state}
\end{equation}
where $\hat{x}$ is a $2n\times 1$ vector of canonical quadrature operators, $\mu$ is a $2n\times 1$ real mean vector, $H$ is a $2n\times 2n $ real symmetric positive definite Hamiltonian matrix, and $Z(\mu,H)$ is a normalization factor that ensures the state is normalized (see Section~\ref{sec:background-bosonic-Gaussian-thermal-state} for background). An arbitrary faithful (positive definite) bosonic Gaussian state can be written in this form, and we refer to this particular parameterization in terms of $\mu$ and $H$ as a \textit{bosonic Gaussian thermal state}, due to the fact that $\rho(\mu,H)$ can be understood from the thermodynamic perspective as a thermal state of the quadratic Hamiltonian $\frac{1}{2} (\hat{x} - \mu)^T H (\hat{x} - \mu)$~\cite[Section~3.2]{Ser17}. Note that all bosonic Gaussian states can be recovered from bosonic Gaussian thermal states by taking limits \cite[Appendix~A]{Lami2018}. Since we work in the thermal-state representation, our results apply to mixed states (i.e., faithful bosonic Gaussian states).

Our main concern here is to understand the distance between two nearby bosonic Gaussian thermal states $\rho(\mu,H)$ and $\rho(\mu+d\mu, H+dH)$, denoted by
\begin{equation}
    d(\rho(\mu,H),\rho(\mu+d\mu, H+dH)), \label{eq:dist-meas-gen}
\end{equation}
where $d$ is a distance measure, $d\mu$ represents an infinitesimal change in the mean vector $\mu$, and $dH$ represents an infinitesimal change in the Hamiltonian matrix $H$. This problem is well known to be within the scope of information geometry, and this distance can be characterized in terms of various generalizations of the notion of classical Fisher information matrix (see~\cite{Bengtsson2006,Liu2019,Sidhu2020,Jarzyna2020,Meyer2021fisherinformationin,sbahi2022provablyefficientvariationalgenerative,scandi2024quantumfisherinformationdynamical} for various reviews). For the classical case, there is one unique Fisher information matrix that characterizes the distance~\cite{Cencov1978}. However, more interestingly for the quantum case, there are many different generalizations of the Fisher information matrix~\cite{Petz1996}, depending on the distance  measure that one chooses in~\eqref{eq:dist-meas-gen}.

\subsection{Main results}

In this paper, we adopt the parameterization in~\eqref{eq:intro-bosonic-Gaussian-thermal-state} and evaluate different quantum generalizations of the Fisher information matrix with respect to this parameterization. 
We derive our formulas from first principles, starting from the recent developments of~\cite{patel2024naturalgradient} and \cite{Wilde2025}, which found analytical expressions for the Fisher--Bures, Kubo--Mori, and $\alpha$-$z$ information matrices of general thermal states. Our main results are as follows:
\begin{enumerate}
    \item formulas for the Fisher--Bures, Kubo--Mori, and $\alpha$-$z$ information matrices of bosonic Gaussian thermal states (Theorem~\ref{thm:main}),
    \item formulas for the  derivatives of bosonic Gaussian thermal states (Theorem~\ref{thm:derivative-bosonic-Gaussian}),
    \item formulas for the symmetric logarithmic derivatives of bosonic Gaussian thermal states (Theorem~\ref{thm:SLD-bosonic-G}).
\end{enumerate}
\noindent Although there is already prior work on deriving the Fisher--Bures information matrices of bosonic Gaussian states (see Section~\ref{sec:lit-rev}),
we believe our derivations are insightful, and we suspect that all of our results listed above will be useful in quantum machine learning algorithms that employ bosonic Gaussian thermal states as an ansatz.

One of the main applications of our formulas is in establishing fundamental limitations on how well one can estimate the parameters $\mu$ and $H$ when given samples of an unknown state $\rho(\mu,H)$ of the form in~\eqref{eq:intro-bosonic-Gaussian-thermal-state}. Indeed, the multiparameter quantum Cramer--Rao bound states that the following matrix inequality holds
for a general parameterized family $(\sigma(\theta))_{\theta\in\mathbb{R}^{J}
}$:
\begin{equation}
\operatorname{Cov}^{(n)}(\hat{\theta},\theta)\geq\frac{1}{n}\left[  I^{\operatorname{FB}
}(\theta)\right]  ^{-1},\label{eq:Cramer--Rao-multiple}
\end{equation}
where $n\in\mathbb{N}$ is the number of copies of the state $\sigma(\theta)$
available, $\hat{\theta}$ is an estimate of the parameter vector~$\theta$, the matrix $I^{\operatorname{FB}}(\theta)$ denotes the Fisher--Bures
information matrix (defined later in~\eqref{eq:qfim-explicit}), and the covariance matrix
$\operatorname{Cov}^{(n)}(\hat{\theta},\theta)$ measures errors in estimation and is
defined in terms of its matrix elements as
\begin{multline}
\lbrack\operatorname{Cov}^{(n)}(\hat{\theta},\theta)]_{k,\ell}\coloneqq\\
\sum_{m}\operatorname{Tr}[M_{m}^{(n)}\sigma(\theta)^{\otimes n}](\hat{\theta
}_{k}(m)-\theta_{k})(\hat{\theta}_{\ell}(m)-\theta_{\ell}).
\label{eq:cov-mat-def}
\end{multline}
In the above, $(M_{m}^{(n)})_{m}$ is an arbitrary positive operator-value measure used for estimation (satisfying $M_{m}^{(n)}\geq0$ for all $m$ and $\sum_{m}M_{m}^{(n)}=I^{\otimes n}$).
In general, this measurement acts collectively on all $n$ copies of the
state $\sigma(\theta)^{\otimes n}$. Additionally,
\begin{equation}
\hat{\theta}(m)\coloneqq(\hat{\theta}_{1}(m),\hat{\theta}_{2}(m),\ldots
,\hat{\theta}_{J}(m))
\label{eq:parameter-estimate-function}
\end{equation}
is a function that maps the measurement outcome $m$ to an estimate
$\hat{\theta}(m)$ of the parameter vector $\theta$. One can formulate a scalar performance metric
from~\eqref{eq:Cramer--Rao-multiple} by setting $W$ to be a positive
semi-definite \textquotedblleft weight\textquotedblright\ matrix and taking
the trace of both sides of~\eqref{eq:Cramer--Rao-multiple} with respect to~$W$,
which leads to
\begin{equation}
\operatorname{Tr}[W\operatorname{Cov}^{(n)}(\hat{\theta},\theta)]\geq\frac{1}
{n}\operatorname{Tr}\left[W\, I^{\operatorname{FB}}(\theta)^{-1}\right].
\label{eq:scalar-Cramer--Rao-from-matrix-one}
\end{equation}
Furthermore, under certain conditions, the inequality is tight up to a factor of two (see~\cite[Eqs.~(15) \& (17)]{Carollo2019},~\cite{Carollo2020}, and~\cite[Theorem~9]{Tsang2020}), in the sense that
\begin{equation}
\operatorname{Tr}[W\operatorname{Cov}^{(n)}(\hat{\theta},\theta)]\leq  \frac{2}
{n}\operatorname{Tr}[W\, I^{\operatorname{FB}}(\theta)^{-1}],    
\end{equation}
so that the Fisher--Bures information matrix is not only foundational in the sense of information geometry, but it is also essential in quantum estimation theory.

Due to the inequalities in~\eqref{eq:Cramer--Rao-multiple} and~\eqref{eq:scalar-Cramer--Rao-from-matrix-one}, our results thus apply to establishing fundamental limits on the performance of unbiased estimators $\hat{\mu}$ and $\hat{H}$ of the mean vector $\mu$ and Hamiltonian matrix $H$, respectively, of a bosonic Gaussian thermal state. 
Our results then are related to but complementary to the problems of learning bosonic Gaussian states~\cite{mele2024learningquantumstatescontinuous,bittel2024optimalestimatestracedistance,fanizza2024efficienthamiltonianstructuretrace}, estimating parameters of thermal states~\cite{GarciaPintos2024,abiuso2024fundamentallimitsmetrologythermal}, and learning thermal states~\cite{Anshu2021}. For our application, the performance metric is indeed a weighted covariance matrix of the form in~\eqref{eq:cov-mat-def}, whereas the learning applications employ different performance metrics (however, see~\cite[Appendix~VIII]{GarciaPintos2024} for a discussion of the relation between the scenarios).

As additional contributions, we establish formulas for the symmetric logarithmic derivative and the derivative of a bosonic Gaussian thermal state, again with respect to its mean vector and Hamiltonian matrix. The first formula should be helpful in determining optimal strategies for single-parameter estimation of bosonic Gaussian thermal states, due to it being known that the symmetric logarithmic derivative is an optimal observable to measure~\cite{braunsteincaves1994} (see also \cite[Eq.~(108)]{Sidhu2020});  we leave the full exploration of implementing such a strategy as a direction for future work. The second aforementioned formula is related to the recent developments in~\cite{bittel2024optimalestimatestracedistance}, where the derivative of a Gaussian state was established in terms of its mean vector and covariance matrix. It is possible that our latter formula could have applications related to those presented in~\cite{bittel2024optimalestimatestracedistance}, and we also leave the full exploration of this for future work. 

Another possible application of our findings, already mentioned above, is in the area of quantum machine learning~\cite{biamonte2017quantum}. Indeed, if one chooses a bosonic Gaussian thermal state as an optimization ansatz, then our formulas for the derivative of a bosonic Gaussian thermal state should be helpful in formulating gradient descent algorithms for optimization tasks. Moreover, our formulas for the Fisher--Bures, Kubo--Mori, and $\alpha$-$z$ information matrices should be helpful in formulating a generalization of gradient descent called natural gradient descent~\cite{amari1998natural,Stokes2020quantumnatural}, as considered recently for general thermal states~\cite{patel2024naturalgradient}. We also leave the full exploration of these directions for future work, with our main focus here being on the derivation of the various aforementioned formulas. 

\subsection{Comparison with prior work}

\label{sec:lit-rev}

Here we briefly recall some prior papers closely related to our findings here, and we also indicate the distinction between our findings and this previous work. A particular quantum generalization of Fisher information of bosonic Gaussian states, called the Fisher--Bures information here, has been studied previously~\cite{monras2013phasespaceformalismquantum,Gao2014} (see also \cite{Marian2016} for even earlier results); we include a detailed mathematical comparison of how the QFI matrix is parametrized by first and second moments in Appendix~\ref{sec:comparison}. In these works, the authors adopted a parameterization of bosonic Gaussian states different from that in~\eqref{eq:intro-bosonic-Gaussian-thermal-state}, and they established formulas for the Fisher--Bures information matrix elements with respect to this parameterization. Namely, they derived a formula for the Fisher--Bures information matrix based on the mean vector and the covariance matrix of the state.
See also~\cite[Eq.~(28)]{Banchi2015}, \cite[Eq.~(8)]{Safranek2018}, \cite[Eq.~(20)]{Nichols2018}, \cite[Theorem~2.10]{Liu2019}, and \cite[Eq.~(2.13)]{Huang2024} for alternate presentations.

The Fisher--Bures information matrices of different parameterizations are related by a congruence transformation expressed in terms of the Jacobian matrix~\cite[Proposition~2.1]{Liu2019}. As such, in principle, the formulas for the Fisher--Bures information matrix presented in our paper can be derived by using this fact and by starting from  the formulas previously derived in~\cite{monras2013phasespaceformalismquantum,Gao2014}, while employing the Jacobian of the transformation that expresses the covariance matrix $V$ as a function of the Hamiltonian matrix $H$ (see~\eqref{eq:V-of-H} for this transformation). However, here we do not follow this approach; furthermore, it is not clear to us whether we would have arrived at the particular formulas presented here by following the aforementioned approach.

A very similar parameterization to that given in \eqref{eq:intro-bosonic-Gaussian-thermal-state}, which is our main focus here, has been adopted and studied in earlier work~\cite{Jiang2014}. 
The most important distinction between our development here and that from~\cite{Jiang2014} is that we make use of the particular derivative formula in Proposition~\ref{prop:deriv}, which was not used in~\cite{Jiang2014}.  This formula was recently shown to be essential in formulating quantum algorithms for quantum Boltzmann machine learning \cite{patel2024quantumboltzmann,patel2024naturalgradient}, which include gradient descent and natural gradient descent.  Although we do not explore this application in detail here, we strongly suspect our formulas derived here will play a similarly essential role in quantum machine learning when using bosonic Gaussian thermal states as an ansatz.

The equations central to this work are
\begin{itemize}
\item the thermal-state derivative Eq.~\eqref{eq:thermal-state-deriv},
\item the QFI matrix element Eq.~\eqref{eq:FB-info-mat-gen-thermal} and \eqref{eq:FB-info-mat-gen-thermal-alt},
\item the bosonic Gaussian properties in Eqs.~\eqref{eq:derivs-Gaussian-Ham}-\eqref{eq:derivs-Gaussian-Ham-2},
\item the main results in Theorem  \ref{thm:main}.
\end{itemize}

Finally, let us note that we have also derived formulas for the Kubo--Mori and $\alpha$-$z$ information matrix elements of bosonic Gaussian thermal states, which was not explored in any prior paper, to the best of our knowledge.   An appealing feature of our derivation is that it is essentially the same as  the approach taken for the derivation of the Fisher--Bures information matrix elements, a fact noted in the statement of Theorem~\ref{thm:main}. The Kubo--Mori  information matrix also has applications in quantum machine learning \cite{patel2024naturalgradient}, and we leave the full exploration of this direction, for bosonic Gaussian thermal states, to future work.

\subsection{Paper organization}

The rest of our paper is organized as follows. Section~\ref{sec:prelim-not} establishes notation and provides some basic background on bosonic Gaussian states and the Fisher--Bures, Kubo--Mori, and $\alpha$-$z$ information matrices. Section~\ref{sec:review-gen-thermal-states} reviews and slightly extends some recent results from~\cite{patel2024naturalgradient}, regarding the derivatives of parameterized thermal states, their Fisher--Bures, Kubo--Mori, and $\alpha$-$z$ information matrices, and their symmetric logarithmic derivatives. Our main results are then presented in Sections~\ref{sec:main-FB-KM-bosonic-G}, \ref{sec:derivative-bosonic-G}, and~\ref{sec:SLD-bosonic-G}, regarding bosonic Gaussian thermal states and their Fisher--Bures, Kubo--Mori, and $\alpha$-$z$ information matrices, their derivatives, and their  symmetric logarithmic derivatives. In Section~\ref{sec:example}, we evaluate the expressions for the Fisher--Bures information matrix for a general single-mode, zero-mean bosonic Gaussian state. Finally, in Section~\ref{sec:conclusion}, we conclude with a brief summary and suggestions for future research.

\section{Preliminaries and notation}

\label{sec:prelim-not}

\subsection{Bosonic Gaussian thermal states}

\label{sec:background-bosonic-Gaussian-thermal-state}

We begin by briefly reviewing fundamentals of bosonic Gaussian thermal states. Our
review applies to $n$-mode  states, where $n$ is some fixed
positive integer. Let $\hat{x}_{j}$ denote each quadrature operator ($2n$ of
them for an $n$-mode state), and let
\begin{align}
\hat{x} &  \coloneqq\left[  \hat{q}_{1},\ldots,\hat{q}_{n},\hat{p}_{1}
,\ldots,\hat{p}_{n}\right]    \equiv\left[  \hat{x}_{1},\ldots,\hat{x}_{2n}\right]
\end{align}
denote the vector of quadrature operators, so that the first~$n$ entries
correspond to position-quadrature operators and the last~$n$ to
momentum-quadrature operators. The quadrature operators satisfy the canonical
commutation relations:
\begin{equation}
\left[  \hat{x}_{j},\hat{x}_{k}\right]  =i\Omega_{j,k}
,\label{eq:symplectic-form}
\end{equation}
where
\begin{equation}
\Omega=
\begin{bmatrix}
0 & 1\\
-1 & 0
\end{bmatrix}
\otimes I_{n}
\end{equation}
and $I_{n}$ is the $n\times n$ identity matrix. We also take the annihilation
operator $\hat{a}\coloneqq\left(  \hat{q}+i\hat{p}\right)  /\sqrt{2}$.

Let
$\rho$ be a Gaussian state, with the mean-vector entries $\mu_{j}
\coloneqq\left\langle \hat{x}_{j}\right\rangle _{\rho}$ for $j \in \{1,\ldots, 2n\}$, and let
$\mu\coloneqq\left(  \mu_{1},\ldots,\mu_{2n}\right)  $ denote the mean vector.
The entries of the $2n\times2n$ Wigner function covariance matrix $V$\ of
$\rho$ are given by
\begin{equation}
V_{j,k}\coloneqq\frac{1}{2}\left\langle \left\{  \hat{x}^c_{j},\hat
{x}^c_{k}\right\}  \right\rangle _{\rho},\label{eq:covariance-matrices}
\end{equation}
where $\{A,B\} \coloneqq AB+BA$ denotes the anticommutator and
\begin{equation}
    \hat{x}^c_{j} \equiv \hat{x}_{j}-\mu_{j}. \label{eq:centered-quad-op}
\end{equation}
A $2n\times2n$ matrix $S$ is symplectic if it preserves the symplectic form:
$S\Omega S^{T}=\Omega$. According to Williamson's theorem~\cite{W36}, there is
a diagonalization of the covariance matrix $V$ of the form
\begin{equation}
V=S\left(
D\oplus D\right)  S^{T},
\label{eq:cov-mat-S-D}
\end{equation}
where $S$ is a symplectic matrix and
$D\coloneqq\operatorname{diag}(\nu_{1},\ldots,\nu_{n})$ is a diagonal matrix
of symplectic eigenvalues such that $\nu_{i}\geq1/2$ for all $i\in\left\{
1,\ldots,n\right\}  $. For a faithful bosonic Gaussian state, all of these latter inequalities are strict (i.e., $\nu_{i}>1/2$ for all $i\in\left\{
1,\ldots,n\right\}  $).

We can write a  bosonic Gaussian thermal state $\rho(\mu,H)$ in the following
exponential form~\cite{Chen2005,K06,H10,H13book,Banchi2015}:
\begin{equation}
\rho(\mu,H)\coloneqq\frac{1}{Z(\mu,H)}\exp\!\left[  -\frac{1}{2}(\hat{x}
-\mu)^{T}H(\hat{x}-\mu)\right],  \label{eq:bosonic-Gaussian-thermal-state}
\end{equation}
where $H$ is a $2n\times 2n $ real symmetric positive definite Hamiltonian matrix, 
\begin{align}
Z(\mu,H)  &  \coloneqq\sqrt{\det(V+i\Omega/2)},\\
V & = \operatorname{coth}\!\left(\frac{i\Omega H}{2}\right) \frac{i\Omega}{2}. \label{eq:V-of-H}
\end{align}
Note that $V$ is indeed the $2n\times 2n $ covariance matrix defined in~\eqref{eq:covariance-matrices}, and~\eqref{eq:V-of-H} expresses $V$ in terms of $H$, with the inverse relationship given by
\begin{align}
H & =2i\Omega\operatorname{arcoth}(2iV\Omega),
\end{align}
where $\operatorname{arcoth}(x)\coloneqq\frac{1}{2}\ln\!\left(  \frac{x+1}
{x-1}\right)  $ with domain $\left(  -\infty,-1\right)  \cup\left(
1,+\infty\right)  $.
Here we exclusively consider the subset $\left(
1,+\infty\right)  $ of the domain. 
We can also write \cite[Appendix~A]{WTLB17}
\begin{equation}
    H    =-2\Omega S\left[  \operatorname{arcoth}(2D)\right]  ^{\oplus
2}  S  ^{T}\Omega \label{eq:G_rho},
\end{equation}
 making use of $S$ and $D$ from~\eqref{eq:cov-mat-S-D}.

Note that the matrix inequality $V+\frac{i\Omega}{2}\geq 0$ holds for a legitimate covariance matrix, and the inequality $V+\frac{i\Omega}{2}> 0$ holds for a faithful bosonic Gaussian state. The latter inequality is equivalent to $H > 0$ (i.e., $H$ being a positive definite real matrix).

Let us note the following alternative forms for the bosonic
Gaussian thermal state $\rho(\mu,H)$:
\begin{align}
&  \rho(\mu,H)\nonumber\\
&  =\frac{1}{Z(\mu,H)}\exp\!\left[  -\frac{1}{2}\sum_{i,j=1}^{2n}h_{ij}
\hat{x}^c_{i}\hat{x}^c_{j}\right] \\
&  =\frac{1}{Z(\mu,H)}\exp\!\left[  -\frac{1}{2}\sum_{i,j=1}^{2n}h_{ij}
\frac{\left\{  \hat{x}^c_{i},\hat{x}_{j}^c\right\}}{2}  
\right]  , \label{eq:sym-exp-form}
\end{align}
where the matrix elements of $H$ are denoted by $h_{i,j}$ and we made use of~\eqref{eq:centered-quad-op}. The
$H$ and $V$ matrices are each symmetric, with the former fact being critical to arriving at~\eqref{eq:sym-exp-form}.

\subsection{Fisher--Bures, Kubo--Mori, and $\alpha$-$z$ information matrices}

We now briefly review the Fisher--Bures, Kubo--Mori, and $\alpha$-$z$ information matrices.
Consider a parameterized family $\left(  \sigma(\theta)\right)  _{\theta
\in\mathbb{R}^{J}}$ of quantum states, and suppose for simplicity that each
state $\sigma(\theta)$ is positive definite.

The Bures distance of two general states $\omega$ and $\tau$~\cite{bures1969extension} is defined as
\begin{align}
d_{B}(\omega,\tau)  &  \coloneqq \sqrt{2\left(  1-\sqrt{F}(\omega,\tau
)\right)  },
\end{align}
where the fidelity $F(\omega,\tau)$ is defined as~\cite{Uhlmann1976}
\begin{equation}
F(\omega,\tau)\coloneqq\left\Vert \sqrt{\omega}\sqrt{\tau}\right\Vert _{1}
^{2}.
\end{equation}
It is well known that the infinitesimal squared line element between two
members of the parameterized family is~\cite[Eq.~(77)]{Liu2019}
\begin{equation}
d_{B}^{2}(\sigma(\theta),\sigma(\theta+d\theta))=\frac{1}{4}\sum_{i,j=1}
^{J}I_{ij}^{\operatorname{FB}}(\theta)\ d\theta_{i}\ d\theta_{j},
\label{eq:Bures-Fisher-metric}
\end{equation}
where $I_{ij}^{\operatorname{FB}}(\theta)$ is a matrix element of a Riemannian
metric tensor called the Fisher--Bures information matrix, given explicitly by
the following formula~\cite[Eqs.~(126)]{Sidhu2020}:
\begin{align}
I_{ij}^{\operatorname{FB}}(\theta)  &  \coloneqq \sum_{k,\ell}\frac{2}
{\lambda_{k}+\lambda_{\ell}} \langle k|\partial_{i}\sigma(\theta)|\ell
\rangle\!\langle\ell|\partial_{j}\sigma(\theta)|k\rangle
\label{eq:qfim-explicit}.
\end{align}
In the above, we make use of the shorthand
\begin{equation}
\partial_{i}\equiv\partial_{\theta_{i}}
\end{equation}
and a spectral decomposition $\rho(\theta)=\sum_{k}\lambda_{k}|k\rangle
\!\langle k|$, where in the latter notation we have suppressed the dependence
of the eigenvalue $\lambda_{k}$ and the eigenvector $|k\rangle$ on the
parameter vector $\theta$.

An alternative metric that one can consider is that induced by the quantum
relative entropy,  defined for states $\omega$ and $\tau$ as
\cite{umegaki1962ConditionalExpectationOperator}
\begin{equation}
D(\omega\Vert\tau)\coloneqq \sup_{\varepsilon > 0} \operatorname{Tr}[\omega(\ln\omega-\ln(\tau + \varepsilon I))].
\label{eq:rel-ent-def}
\end{equation}
 The quantum relative entropy
is a fundamental distinguishability measure in quantum information theory, due
to it obeying the data-processing inequality~\cite{Lindblad1975}\ and having
an operational meaning in the context of quantum hypothesis testing
\cite{hiai1991ProperFormulaRelative,nagaoka2000StrongConverseSteins}.

Following reasoning similar to that leading to~\eqref{eq:Bures-Fisher-metric},
the infinitesimal line element between two members of the parameterized family
$\left(  \sigma(\theta)\right)  _{\theta\in\mathbb{R}^{J}}$ is
\begin{equation}
D(\sigma(\theta)\Vert\sigma(\theta+d\theta))=\frac{1}{2}\sum_{i,j=1}^{J}
I_{ij}^{\operatorname{KM}}(\theta)\ d\theta_{i}\ d\theta_{j},
\label{eq:KM-metric-def}
\end{equation}
where $I_{ij}^{\operatorname{KM}}(\theta)$ is a matrix element of a Riemannian
metric tensor called the Kubo--Mori information matrix, given explicitly by
the following formulas~\cite[Eqs.~(B9), (B12), (B22), (B23)]
{sbahi2022provablyefficientvariationalgenerative}:
\begin{align}
I_{ij}^{\operatorname{KM}}(\theta)  &  \coloneqq\sum_{k,\ell}
c_{\operatorname{KM}}(\lambda_{k},\lambda_{\ell})\langle k|\partial_{i}
\sigma(\theta)|\ell\rangle\!\langle\ell|\partial_{j}\sigma(\theta)|k\rangle,
\label{eq:KM-inf-matrix-explicit}
\end{align}
where
\begin{equation}
c_{\operatorname{KM}}(x,y)\coloneqq\left\{
\begin{array}
[c]{cc}
\frac{1}{x} & \text{if }x=y\\
\frac{\ln x-\ln y}{x-y} & \text{if }x\neq y
\end{array}
\right.  .
\end{equation}

Generalizing the relative entropy is the $\alpha$-$z$ R\'enyi relative entropy \cite{Audenaert2015}, defined for states $\omega$ and $\tau$ and $\alpha\in (0,1)\cup (1,\infty)$ and $z>0$ as
\begin{equation}
    D_{\alpha,z}(\omega \| \tau) \coloneqq \sup _{\varepsilon>0}\frac{1}{\alpha-1}\ln\operatorname{Tr}\!\left[\left(\tau_{\varepsilon}^{\frac{1-\alpha}{2z}}\omega ^{\frac{\alpha}{z}}\tau_{\varepsilon}^{\frac{1-\alpha}{2z}}\right)^{z}\right],
\end{equation}
where $\tau_{\varepsilon} \coloneqq \tau + \varepsilon I$. This quantity generalizes many other distinguishability measures in quantum information theory and obeys the data-processing inequality for certain values of $\alpha$ and $z$ \cite[Theorem~1.2]{Zhang2020}.

Again following reasoning similar to that leading to~\eqref{eq:Bures-Fisher-metric},
the infinitesimal line element between two members of the parameterized family
$\left(  \sigma(\theta)\right)  _{\theta\in\mathbb{R}^{J}}$ is
\begin{equation}
D_{\alpha,z}(\sigma(\theta)\Vert\sigma(\theta+d\theta))=\frac{\alpha}{2}\sum_{i,j=1}^{J}
I_{ij}^{\alpha,z}(\theta)\ d\theta_{i}\ d\theta_{j},
\label{eq:a-z-metric-def}
\end{equation}
where $I_{ij}^{\alpha,z}(\theta)$ is a matrix element of a Riemannian
metric tensor called the $\alpha$-$z$ information matrix, given explicitly by
the following formula~\cite[Theorem~10]
{Wilde2025}:
\begin{align}
I_{ij}^{\alpha,z}(\theta)  &  \coloneqq\sum_{k,\ell}
c_{\alpha,z}(\lambda_{k},\lambda_{\ell})\langle k|\partial_{i}
\sigma(\theta)|\ell\rangle\!\langle\ell|\partial_{j}\sigma(\theta)|k\rangle,
\label{eq:a-z-inf-matrix-explicit}
\end{align}
where
\begin{multline}
c_{\alpha,z}(x,y)\coloneqq\\
\left\{
\begin{array}
[c]{cc}
\frac{1}{x} & \text{if }x=y\\[1em]
\frac{z}{\alpha\left(1-\alpha\right)}\left(\frac{x^{\frac{1-\alpha}{z}}-y^{\frac{1-\alpha}{z}}}{x-y}\right)\left(\frac{x^{\frac{\alpha}{z}}-y^{\frac{\alpha}{z}}}{x^{\frac{1}{z}}-y^{\frac{1}{z}}}\right) & \text{if }x\neq y
\end{array}
\right.  .
\end{multline}

\section{General thermal
states}

\label{sec:review-gen-thermal-states}

In this section, we recall some of the recent results from~\cite{patel2024naturalgradient}, and while doing so, we slightly generalize them such that they are suitable for our purposes here. Let
\begin{equation}
\theta\mapsto G(\theta)
\end{equation}
be an operator-valued differentiable function from $\theta\in\mathbb{R}^{J}$
to the set of Hermitian operators acting on a Hilbert space $\mathcal{H}$.
Then the following defines an operator-valued differentiable function from
$\theta\in\mathbb{R}^{J}$ to the set of thermal states acting on $\mathcal{H}
$:
\begin{equation}
\theta\mapsto\rho(\theta)\coloneqq\frac{e^{-G(\theta)}}{Z(\theta
)}.\label{eq:thermal-state}
\end{equation}
Let $\Phi_{\theta}$ denote the following quantum channel:
\begin{align}
\Phi_{\theta}(X) &  \coloneqq\int_{\mathbb{R}}dt\ p(t)\ e^{iG(\theta
)t}Xe^{-iG(\theta)t},\label{eq:Phi-channel}\\
p(t) &  \coloneqq\frac{2}{\pi}\ln\left\vert \coth\!\left(  \frac{\pi t}
{2}\right)  \right\vert ,\label{eq:high-peak-tent}
\end{align}
where $p(t)$ is a probability density function known as the high-peak tent
probability density~\cite{patel2024quantumboltzmann}.

\subsection{Derivative of a thermal state}

The following proposition presents an expression for the derivative of a thermal state, slightly generalizing that presented previously in~\cite[Eq.~(9)]
{Hastings2007},~\cite[Proposition~20]{Anshu2020arXiv}, and~\cite[Lemma~5]
{Coopmans2024}.

\begin{proposition}
[Derivative of a thermal state]\label{prop:deriv} For $\rho(\theta)$ defined
as in~\eqref{eq:thermal-state}, the following equality holds:
\begin{equation}
\partial_{j}\rho(\theta)=-\frac{1}{2}\left\{  \Phi_{\theta}(\partial
_{j}G(\theta)),\rho(\theta)\right\}  +\rho(\theta)\left\langle \partial
_{j}G(\theta)\right\rangle _{\rho(\theta)},
\label{eq:thermal-state-deriv}
\end{equation}
where the channel $\Phi_{\theta}$ is defined in~\eqref{eq:Phi-channel}.

\end{proposition}

\begin{proof}
See Appendix~\ref{app:deriv-proof}.
\end{proof}

\subsection{Fisher--Bures, Kubo--Mori, and $\alpha$-$z$ information matrices}

Here we present expressions for the matrix elements of the Fisher--Bures and Kubo--Mori information matrices of thermal states, generalizing those presented recently in \cite[Theorems~1 and 2]{patel2024naturalgradient}.

\begin{proposition}
\label{prop:FB-KM-info-mats} The Fisher--Bures and Kubo--Mori Fisher
information matrix elements of the parameterized family $\left(  \rho
(\theta)\right)  _{\theta\in\mathbb{R}^{J}}$, with $\rho(\theta)$ defined as
in~\eqref{eq:thermal-state}, are as follows:
\begin{align}
I_{ij}^{\operatorname{FB}}(\theta)  &  =\frac{1}{2}\left\langle \left\{
\Phi_{\theta}(\partial_{i}G(\theta)),\Phi_{\theta}(\partial_{j}G(\theta
))\right\}  \right\rangle _{\rho(\theta)}\nonumber\\
&  \qquad-\left\langle \partial_{i}G(\theta)\right\rangle _{\rho(\theta
)}\left\langle \partial_{j}G(\theta)\right\rangle _{\rho(\theta)},\label{eq:FB-info-mat-gen-thermal}\\ 
I_{ij}^{\operatorname{KM}}(\theta)  &  =\frac{1}{2}\left\langle \left\{
\partial_{i}G(\theta),\Phi_{\theta}(\partial_{j}G(\theta))\right\}
\right\rangle _{\rho(\theta)}\nonumber\\
&  \qquad-\left\langle \partial_{i}G(\theta)\right\rangle _{\rho(\theta
)}\left\langle \partial_{j}G(\theta)\right\rangle _{\rho(\theta)},
\label{eq:KM-info-mat-gen-thermal}
\end{align}
where the channel $\Phi_{\theta}$ is defined in~\eqref{eq:Phi-channel}.

\end{proposition}

\begin{proof}
See Appendix~\ref{app:FB-KM-proofs}.
\end{proof}

\medskip

By defining the following channel:
\begin{align}
\Psi_{\theta}(X)  &  \coloneqq\int_{\mathbb{R}}dt\ q(t)\ e^{iG(\theta
)t}Xe^{-iG(\theta)t},\label{eq:Psi-ch}\\
q(t)  &  \coloneqq\int_{\mathbb{R}}d\tau\ p(\tau)p(t+\tau),
\label{eq:convolve-high-peak-tent}
\end{align}
we have the following alternative form for the Fisher--Bures information
matrix elements of thermal states, which is helpful in arriving at simpler expressions later on for bosonic Gaussian thermal states:

\begin{proposition}
\label{prop:FB-info-mats-alt} The Fisher--Bures information matrix elements of
the parameterized family $\left(  \rho(\theta)\right)  _{\theta\in
\mathbb{R}^{J}}$, with $\rho(\theta)$ defined as in~\eqref{eq:thermal-state},
are as follows:
\begin{multline}
I_{ij}^{\operatorname{FB}}(\theta)=\frac{1}{2}\left\langle \left\{
\partial_{i}G(\theta)),\Psi_{\theta}(\partial_{j}G(\theta))\right\}
\right\rangle _{\rho(\theta)}\label{eq:FB-info-mat-gen-thermal-alt}\\
-\left\langle \partial_{i}G(\theta)\right\rangle _{\rho(\theta)}\left\langle
\partial_{j}G(\theta)\right\rangle _{\rho(\theta)},
\end{multline}
where $\Psi_{\theta}$ is the channel defined in~\eqref{eq:Psi-ch}.

\end{proposition}

\begin{proof}
See Appendix~\ref{app:alt-form-FB-inf-mat}.
\end{proof}

\medskip 

Generalizing both of these expressions is that for the $\alpha$-$z$ information matrix. The following proposition is a consequence of \cite[Theorem~32]{Wilde2025} and a similar generalization used to prove Proposition~\ref{prop:FB-KM-info-mats} (i.e., replacing $G_i$ with $\partial_i G(\theta)$ throughout).

\begin{proposition}
    The $\alpha$-$z$ information matrix elements of the parameterized family $\left(  \rho(\theta)\right)  _{\theta\in
\mathbb{R}^{J}}$, with $\rho(\theta)$ defined as in~\eqref{eq:thermal-state},
are as follows for all $\alpha\in\left(0,1\right)$ and $z>0$:
\begin{multline}
I^{\alpha,z}(\theta)_{i,j}=\frac{1}{2}\left\langle \left\{ \Phi_{q_{\alpha,z},\theta}\!\left(\partial_i G(\theta)\right),\partial_j G(\theta)\right\} \right\rangle _{\rho(\theta)}\\
-\left\langle \partial_i G(\theta)\right\rangle _{\rho(\theta)}\left\langle \partial_j G(\theta)\right\rangle _{\rho(\theta)},
\end{multline}
where the channel $\Phi_{q_{\alpha,z},\theta}$ is given by
\begin{equation}
\Phi_{q_{\alpha,z},\theta}\!\left(X\right)\coloneqq\int_{-\infty}^{\infty}dt\:q_{\alpha,z}(t)e^{-itG(\theta)}Xe^{itG(\theta)},
\end{equation}
and $q_{\alpha,z}$ is the following probability density function:
\begin{equation}
q_{\alpha,z}(t)\coloneqq\left(p*p_{\alpha,z}\right)(t)=\int_{-\infty}^{\infty}d\tau\:p(\tau)p_{\alpha,z}(t-\tau),\label{eq:q-a-z-prob-dens}
\end{equation}
where the probability density function $p$ is defined in \eqref{eq:high-peak-tent} and the probability density function $p_{\alpha,z}$ is
defined on $t\in\mathbb{R}$ as
\begin{align}
p_{\alpha,z}(t) & \coloneqq\frac{z}{2\pi\alpha\left(1-\alpha\right)}\ln\!\left(1+\left(\frac{\sin(\pi\alpha)}{\sinh(\pi zt)}\right)^{2}\right)\label{eq:high-peak-tent-a-z}.
\end{align}
\end{proposition}

\begin{remark}
    Note that $q_{\alpha,z}(t)$ is an even function because it is the convolution of two even functions. This implies that
    \begin{equation}
\Phi_{q_{\alpha,z},\theta}\!\left(X\right)=\int_{-\infty}^{\infty}dt\:q_{\alpha,z}(t)\:e^{itG(\theta)}Xe^{-itG(\theta)}.
\end{equation}
Furthermore, we also have that 
\begin{multline}
I^{\alpha,z}(\theta)_{i,j}=\frac{1}{2}\left\langle \left\{ \partial_i G(\theta),\Phi_{q_{\alpha,z},\theta}\!\left(\partial_j G(\theta)\right)\right\} \right\rangle _{\rho(\theta)}\\
-\left\langle \partial_i G(\theta)\right\rangle _{\rho(\theta)}\left\langle \partial_j G(\theta)\right\rangle _{\rho(\theta)},
\end{multline}
because $G(\theta)$ and $\rho(\theta)$ commute.
\end{remark}

\subsection{Symmetric logarithmic derivative of a thermal state}

 The symmetric logarithmic derivative (SLD) operator $L^{(j)}(\theta)$ corresponding to $\theta_j$ is generally defined for a paramaterized family $(\sigma(\theta))_{\theta \in \mathbb{R}^J}$ as follows~\cite[Eq.~(86)]{Sidhu2020}:
\begin{equation}
    L^{(j)}(\theta) \coloneqq \sum_{k,\ell} \frac{2}{\lambda
_{k}+\lambda_{\ell}}|k\rangle \!\langle k|(\partial_{j}\sigma(\theta
))|\ell\rangle \!\langle \ell | .
\end{equation}
Note that the SLD satisfies the following differential equation~\cite[Eq.~(83)]{Sidhu2020}:
\begin{equation}
    \partial_{j}\sigma(\theta
) = \frac{1}{2} \left\{ L^{(j)}(\theta), \sigma(\theta)\right\}.
\label{eq:SLD-diff-eq}
\end{equation}

For thermal states of the form in~\eqref{eq:thermal-state}, we have the following result, representing a slight generalization of the recent finding of~\cite[Theorem~3]{patel2024naturalgradient}:

\begin{proposition}
    \label{prop:main-sld}For the parameterized family of thermal states in
\eqref{eq:thermal-state}, the SLD operator $L^{(j)}(\theta)$
is given as follows:
\begin{equation}
L^{(j)}(\theta)=-\Phi_{\theta}(\partial_j G(\theta)) + \left\langle
\partial_j G(\theta)\right\rangle _{\rho(\theta)} I,\label{eq:SLD-op}
\end{equation}
where $\Phi_{\theta}$ is the quantum channel defined in~\eqref{eq:Phi-channel}.
\end{proposition}

\begin{proof}
See Appendix~\ref{app:proof-SLD}.
\end{proof}

\medskip 

Observe that the equality in~\eqref{eq:SLD-diff-eq} holds for~\eqref{eq:thermal-state-deriv} and~\eqref{eq:SLD-op}. 

\section{Fisher--Bures, Kubo--Mori, and $\alpha$-$z$ information matrix elements of bosonic
Gaussian thermal states}

\label{sec:main-FB-KM-bosonic-G}

In this section, we present our formulas for the Fisher--Bures and Kubo--Mori information matrices for
the parameterized family $\left(  \rho(\mu,H)\right)  _{\mu,H}$ of bosonic
Gaussian thermal states. Observe that this family is a special case of the
more general family of thermal states in~\eqref{eq:thermal-state}, where
\begin{align}
\theta &  =\left(  \mu,H\right)  ,\label{eq:Gauss-sub-1}\\
G(\theta) &  =\frac{1}{2}(\hat{x}-\mu)^{T}H(\hat{x}-\mu)=\frac{1}{2}(\hat
{x}^{c})^{T}H\hat{x}^{c},\label{eq:Gauss-sub-2}
\end{align}
with $\hat{x}^{c}\equiv\hat{x}-\mu$. As such, we proceed by plugging
into the formulas from Propositions~\ref{prop:FB-KM-info-mats} and
\ref{prop:FB-info-mats-alt} and simplifying, in order to determine expressions for the
Fisher--Bures and Kubo--Mori information matrix elements. In order to do so,
it is necessary to determine the partial derivatives $\frac{\partial}
{\partial\theta_{j}}G(\theta)$. They are as follows:
\begin{align}
\frac{\partial}{\partial h_{k,\ell}}\left(  \frac{1}{2}(\hat{x}-\mu)^{T}
H(\hat{x}-\mu)\right)   &  =\frac{1}{4}\left\{  \hat{x}_{k}-\mu_{k},\hat
{x}_{\ell}-\mu_{\ell}\right\}  \nonumber\\
&  =\frac{1}{4}\left\{  \hat{x}_{k}^{c},\hat{x}_{\ell}^{c}\right\}
,\label{eq:derivs-Gaussian-Ham}\\
\frac{\partial}{\partial\mu_{m}}\left(  \frac{1}{2}(\hat{x}-\mu)^{T}H(\hat
{x}-\mu)\right)   &  =-\sum_{j=1}^{2n}h_{m,j}(\hat{x}_{j}-\mu_{j})\nonumber\\
&  =-\sum_{j=1}^{2n}h_{m,j}\hat{x}_{j}^{c},\label{eq:derivs-Gaussian-Ham-2}
\end{align}
To arrive at~\eqref{eq:derivs-Gaussian-Ham}, we made use of
\eqref{eq:sym-exp-form} and  the shorthand in~\eqref{eq:centered-quad-op}.

When denoting the Fisher--Bures or Kubo--Mori information matrices of bosonic Gaussian thermal states, we adopt a particular indexing scheme, which has to do with the fact that the parameter vector in this case is given by $\theta = (\mu, H)$, with $\mu$ a vector and $H$ a matrix. We adopt the convention of flattening the matrix $H$ to a vector, so that
\begin{equation}
    \theta = (\mu_1, \ldots, \mu_{2n}, h_{1,1}, \ldots, h_{2n,1}, \ldots, h_{1,2n}, \ldots, h_{2n,2n}).
\end{equation}
Then we index the Fisher information matrix elements by $I_{m_1, m_2}$ for the elements of the upper left block, $I_{m, (k,\ell)}$ for the elements of the upper right block, $I_{(k,\ell),m}$ for the elements of the lower left block, and $I_{(k_1, \ell_1),(k_2, \ell_2)}$ for the elements of the lower right block. For example, for a single-mode state, the parameter vector $\theta$ is given by $\theta = (\mu_1, \mu_2, h_{1,1}, h_{2,1}, h_{1,2}, h_{2,2})$ and the Fisher information matrix elements by
{\small
\begin{equation}
    \begin{bmatrix}
        I_{1,1} & I_{1,2} & I_{1,(1,1)} & I_{1,(1,2)} & I_{1,(2,1)} & I_{1,(2,2)}\\
        I_{2,1} & I_{2,2} & I_{2,(1,1)} & I_{2,(1,2)} & I_{2,(2,1)} & I_{2,(2,2)}\\
        I_{(1,1),1} & I_{(1,1),2} & I_{(1,1),(1,1)} & I_{(1,1),(1,2)} & I_{(1,1),(2,1)} & I_{(1,1),(2,2)}\\
        I_{(1,2),1} & I_{(1,2),2} & I_{(1,2),(1,1)} & I_{(1,2),(1,2)} & I_{(1,2),(2,1)} & I_{(1,2),(2,2)}\\
        I_{(2,1),1} & I_{(2,1),2} & I_{(2,1),(1,1)} & I_{(2,1),(1,2)} & I_{(2,1),(2,1)} & I_{(2,1),(2,2)}\\
        I_{(2,2),1} & I_{(2,2),2} & I_{(2,2),(1,1)} & I_{(2,2),(1,2)} & I_{(2,2),(2,1)} & I_{(2,2),(2,2)}
    \end{bmatrix}.
    \label{eq:single_mode_qfi_matrix}
\end{equation}}

With this convention established, we now state our main theorem of this section:

\begin{theorem}
\label{thm:main}The Fisher--Bures information matrix elements of the
parameterized family $\left(  \rho(\mu,H)\right)  _{\mu,H}$, with $\rho
(\mu,H)$ defined as in~\eqref{eq:bosonic-Gaussian-thermal-state}, are as
follows:
\begin{align}
I_{m_{1},m_{2}}^{\operatorname{FB}}(\mu,H) &  =\int_{\mathbb{R}
}dt\ q(t)\left[  HVS(t)H\right]  _{m_{1},m_{2}},\label{eq:cross-term-mean}
\\
I_{m,\left(
k,\ell\right)  }^{\operatorname{FB}}(\mu,H)  &  =I_{\left(  k,\ell\right)  ,m}^{\operatorname{FB}}(\mu,H) =0,\label{eq:zero-cross-terms-nice}
\\
I_{\left(  k_{1},\ell_{1}\right)  ,\left(  k_{2},\ell_{2}\right)
}^{\operatorname{FB}}(\mu,H) &  =\frac{1}{4}\int_{\mathbb{R}}dt\ q(t)W_{k_{1}
,\ell_{1},k_{2},\ell_{2}}(t)\nonumber\\
&  \qquad-\frac{1}{4}V_{k_{1},\ell_{1}}V_{k_{2},\ell_{2}}
,\label{eq:cross-term-Ham-mat}
\end{align}
where $q(t)$ is defined in~\eqref{eq:convolve-high-peak-tent},
\begin{multline}
W_{k_{1},\ell_{1},k_{2},\ell_{2}}(t)\coloneqq V_{k_{1},\ell_{1}}\left[
S(t)^{T}VS(t)\right]  _{k_{2},\ell_{2}}\\
+\left[  VS(t)\right]  _{k_{1},k_{2}}\left[  VS(t)\right]  _{\ell_{1},\ell
_{2}}+\left[  VS(t)\right]  _{k_{1},\ell_{2}}\left[  VS(t)\right]  _{\ell
_{1},k_{2}}\\
-\frac{1}{4}\left[  \Omega S(t)\right]  _{k_{1},k_{2}}\left[  \Omega
S(t)\right]  _{\ell_{1},\ell_{2}}-\frac{1}{4}\left[  \Omega S(t)\right]
_{k_{1},\ell_{2}}\left[  \Omega S(t)\right]  _{\ell_{1},k_{2}}
\end{multline}
and
\begin{equation}
S(t)\coloneqq e^{-H\Omega t}.
\end{equation}
Furthermore, the Kubo--Mori information matrix elements
\begin{align}
& I_{m_{1}
,m_{2}}^{\operatorname{KM}}(\mu,H),\\
& I_{m,\left(  k,\ell\right)  }^{\operatorname{KM}}(\mu,H),\\
& I_{\left(  k,\ell\right)  ,m}^{\operatorname{KM}}(\mu,H), \\ & 
I_{\left(  k_{1}
,\ell_{1}\right)  ,\left(  k_{2},\ell_{2}\right)  }^{\operatorname{KM}}
(\mu,H)
\end{align}
are precisely the same as above, but with
the substitution $q(t)\rightarrow p(t)$, where $p(t)$ is defined in~\eqref{eq:high-peak-tent}. Finally, the $\alpha$-$z$ information matrix elements
\begin{align}
& I_{m_{1}
,m_{2}}^{\alpha,z}(\mu,H),\\
& I_{m,\left(  k,\ell\right)  }^{\alpha,z}(\mu,H),\\
& I_{\left(  k,\ell\right)  ,m}^{\alpha,z}(\mu,H), \\ & 
I_{\left(  k_{1}
,\ell_{1}\right)  ,\left(  k_{2},\ell_{2}\right)  }^{\alpha,z}
(\mu,H)
\end{align}
are precisely the same as those for the Fisher--Bures, but with the substitution $q(t)\rightarrow q_{\alpha,z}(t)$, where $q_{\alpha,z}(t)$ is defined in \eqref{eq:high-peak-tent-a-z}.
\end{theorem}

\begin{proof}
See Appendix~\ref{app:main-thm-proof}.
\end{proof}

\section{Derivative of a bosonic Gaussian thermal state}

\label{sec:derivative-bosonic-G}

Prior work presented a derivative of a bosonic Gaussian state with respect its mean and covariance matrix \cite[Section~B.3]{bittel2024optimalestimatestracedistance}. In this section, we provide an alternative notion of a derivative of a bosonic
Gaussian state, by making direct use of Proposition~\ref{prop:deriv}\ and~\eqref{eq:derivs-Gaussian-Ham}--\eqref{eq:derivs-Gaussian-Ham-2}.

\begin{theorem}
\label{thm:derivative-bosonic-Gaussian}
For a bosonic Gaussian thermal
state $\rho(\mu,H)$, defined as in~\eqref{eq:bosonic-Gaussian-thermal-state},
the following identities hold:
\begin{multline}
\frac{\partial}{\partial h_{k,\ell}}\rho(\mu,H)=\label{eq:deriv-wrt-Ham}\\
-\frac{1}{8}\int_{\mathbb{R}}dt\ p(t)\sum_{k^{\prime},\ell^{\prime}
}S_{k,k^{\prime}}(t)S_{\ell,\ell^{\prime}}(t)\left\{  \left\{  \hat
{x}_{k^{\prime}}^{c},\hat{x}_{\ell^{\prime}}^{c}\right\}  ,\rho(\mu
,H)\right\}  \\
+\frac{V_{k,\ell}}{2}\, \rho(\mu,H)\ ,
\end{multline}
\begin{multline}
\frac{\partial}{\partial\mu_{m}}\rho(\mu,H)=\label{eq:deriv-wrt-mean}\\
\frac{1}{2}\int_{\mathbb{R}}dt\ p(t)\sum_{j,j^{\prime}}h_{m,j}S_{j,j^{\prime}
}(t)\left\{  \hat{x}_{j^{\prime}}^{c},\rho(\mu,H)\right\}  ,
\end{multline}
where $p(t)$ is defined in~\eqref{eq:high-peak-tent}, $\hat{x}_{k}^{c}$ in
\eqref{eq:centered-quad-op}, and
\begin{equation}
S(t)\coloneqq e^{\Omega Ht}.
\end{equation}

\end{theorem}

\begin{proof}
See Appendix~\ref{app:deriv-bos-G-therm}.
\end{proof}

\medskip 

We remark here briefly that one application of this finding is in developing gradient descent algorithms for optimization when using bosonic Gaussian thermal states as an ansatz. In particular, if the objective function is linear in $\rho(\mu,H)$, as is the case for $\operatorname{Tr}[G \rho(\mu,H)]$, where $G$ is a Hamiltonian, then the expressions in \eqref{eq:deriv-wrt-Ham} and \eqref{eq:deriv-wrt-mean} are relevant for developing a gradient descent algorithm that minimizes this objective function. We leave the full exploration of this direction for future work. 

\section{Symmetric logarithmic derivative of a bosonic Gaussian thermal state}

\label{sec:SLD-bosonic-G}

In this section, we provide an explicit form for the symmetric logarithmic derivative of a bosonic Gaussian thermal state, by employing the expressions in Proposition~\ref{prop:main-sld} and~\eqref{eq:derivs-Gaussian-Ham}--\eqref{eq:derivs-Gaussian-Ham-2}. As mentioned in the introduction, it is well known that this observable is optimal for single-parameter estimation~\cite{braunsteincaves1994} (see also \cite[Eq.~(108)]{Sidhu2020}), and so our result here could be useful for this purpose. However, more work is needed to determine how to implement a measurement of this observable.

\begin{theorem}
\label{thm:SLD-bosonic-G}
    For a bosonic Gaussian thermal
state $\rho(\mu,H)$, defined as in~\eqref{eq:bosonic-Gaussian-thermal-state}, the symmetric logarithmic derivatives  are as follows:
    \begin{multline}
        L^{(h_{k,\ell})}(\mu,H)  =  \\
         -\frac{1}{4}\int_{\mathbb{R}}dt\ p(t)\sum_{k^{\prime},\ell^{\prime}
}S_{k,k^{\prime}}(t)S_{\ell,\ell^{\prime}}(t)  \left\{  \hat
{x}_{k^{\prime}}^{c},\hat{x}_{\ell^{\prime}}^{c}\right\} + \frac{V_{k,\ell}}{2} \, I ,
\label{eq:deriv-Ham-mat}
    \end{multline}
    \begin{equation}
        L^{(\mu_{m})}(\mu,H)  =    
    \int_{\mathbb{R}}dt\ p(t)\sum_{j,j^{\prime}}h_{m,j}
S_{j,j^{\prime}}(t)  \hat{x}_{j^{\prime}}^{c},
\label{eq:deriv-mean-vec}
    \end{equation}
    where $p(t)$, $\hat{x}_{k}^{c}$, and $S(t)$ are defined in Theorem~\ref{thm:derivative-bosonic-Gaussian}.
\end{theorem}

\begin{proof}
    The proof is very similar to the proof of Theorem~\ref{thm:derivative-bosonic-Gaussian} (see Appendix~\ref{app:deriv-bos-G-therm}), and so we omit it. 
\end{proof}

\medskip 

Observe that the equality in~\eqref{eq:SLD-diff-eq} holds for~\eqref{eq:deriv-wrt-Ham} and~\eqref{eq:deriv-Ham-mat}, as well as for~\eqref{eq:deriv-wrt-mean} and~\eqref{eq:deriv-mean-vec}.

\section{Example: Squeezed thermal state}

\label{sec:example}

In this section, we examine the Fisher--Bures information matrix elements of a squeezed thermal state, which is a general form for a single-mode, zero-mean bosonic Gaussian state.  A thermal state with mean photon number $\epsilon\geq 0$ has a diagonal covariance matrix of the following form:
\begin{align}
D =
\frac{1}{2} \left(
\begin{array}{cc}
 1 +\epsilon & 0 \\
 0 & 1 +\epsilon \\
\end{array}
\right).
\end{align}
For single-mode squeezing along the momentum quadrature $p$ followed by rotation by an angle $\theta$, the corresponding symplectic matrix $S$ is as follows:
\begin{align}
S &=  \left(
\begin{array}{cc}
 \cos (\theta ) & \sin (\theta ) \\
 -\sin (\theta ) & \cos (\theta ) \\
\end{array}
\right) \left(
\begin{array}{cc}
e^r & 0 \\
 0 & e^{-r} \\
\end{array}
\right).
\end{align}

We can compute the matrix $H$ from the formula \cite{WTLB17}
\begin{equation}
    H =  - 2  \Omega S [\operatorname{arcoth}(2 D)]^{\oplus 2} S^T \Omega
\end{equation}
as follows:
\begin{align}
&H =  \log \! \left(\frac{\epsilon +2}{\epsilon }\right) \times \nonumber \\
&\left[
\begin{array}{cc}
   e^{2 r} \sin ^2(\theta )+ e^{-2 r} \cos ^2(\theta ) &   \sin (2 \theta ) \sinh (2 r) \\
   \sin (2 \theta ) \sinh (2 r) &   e^{-2 r} \sin ^2(\theta )+  e^{2 r} \cos ^2(\theta ) \\
\end{array}
\right].
\end{align}

For a single-mode state, the $H$ matrix has elements
\begin{align}
H = \begin{bmatrix}
h_{xx} & h_{xy} \\
h_{xy} & h_{yy} 
\end{bmatrix},
\end{align}
the parameter vector $\theta$ is given by $\theta = (\mu_1, \mu_2, h_{1,1}, h_{2,1}, h_{1,2}, h_{2,2})$ and the Fisher--Bures information matrix elements by
{\small
\begin{equation} I^{\operatorname{FB}}=
    \begin{bmatrix}
        I_{1,1} & I_{1,2} & 0 & 0 & 0 & 0\\
        I_{2,1} & I_{2,2} & 0 & 0 & 0 & 0\\
        0 & 0 & I_{(1,1),(1,1)} & I_{(1,1),(1,2)} & I_{(1,1),(2,1)} & I_{(1,1),(2,2)}\\
        0 & 0 & I_{(1,2),(1,1)} & I_{(1,2),(1,2)} & I_{(1,2),(2,1)} & I_{(1,2),(2,2)}\\
        0 & 0 & I_{(2,1),(1,1)} & I_{(2,1),(1,2)} & I_{(2,1),(2,1)} & I_{(2,1),(2,2)}\\
        0 & 0 & I_{(2,2),(1,1)} & I_{(2,2),(1,2)} & I_{(2,2),(2,1)} & I_{(2,2),(2,2)}
    \end{bmatrix}.
\end{equation}}

In Fig.~\ref{fig:single_mode_sq}, we plot select elements of the Fisher--Bures information matrix $I^{\operatorname{FB}}$, for the single-mode squeezed thermal state, evaluated at $\theta =0$ and at $\theta = \pi/4$. The python code needed to reproduce this figure is available along with the arXiv posting of our paper. The diagonal entries $I_{(1,1),(1,1)}$ and $I_{(2,2),(2,2)}$ are associated with the Hamiltonian parameters $h_{xx}$ and $h_{yy}$, respectively. In the upper panel, plotted on a linear–log scale versus the squeezing parameter $r$, these two matrix elements appear as straight lines. This behavior is expected: when the displacements $\mu_1,\mu_2 = 0$, the associated matrix elements are proportional to the state’s energy, and that energy in turn is proportional to the corresponding entries of the Hamiltonian matrix $H$. For example, at $r=-1$ the state is strongly squeezed along $\hat x$, so that $h_{xx} \propto e^{-2}$, is being suppressed, and the associated Fisher--Bures matrix element is correspondingly small. More generally, the matrix elements, for a given quadrature, grow linearly with the energy stored in that quadrature, while the mean energy of the state itself grows exponentially with~$r$.

The other elements correspond to the off-diagonal terms and are products of the quadrature operators; hence, they are constant. In the bottom plot, the state is squeezed at $45^\circ$ to the $x$ axis, and we see a minimum in the Fisher--Bures matrix element in the parameters corresponding to the minimum energy state, $r=0$.

\begin{figure}[h!]
\begin{center}
\includegraphics[width=0.5\textwidth]{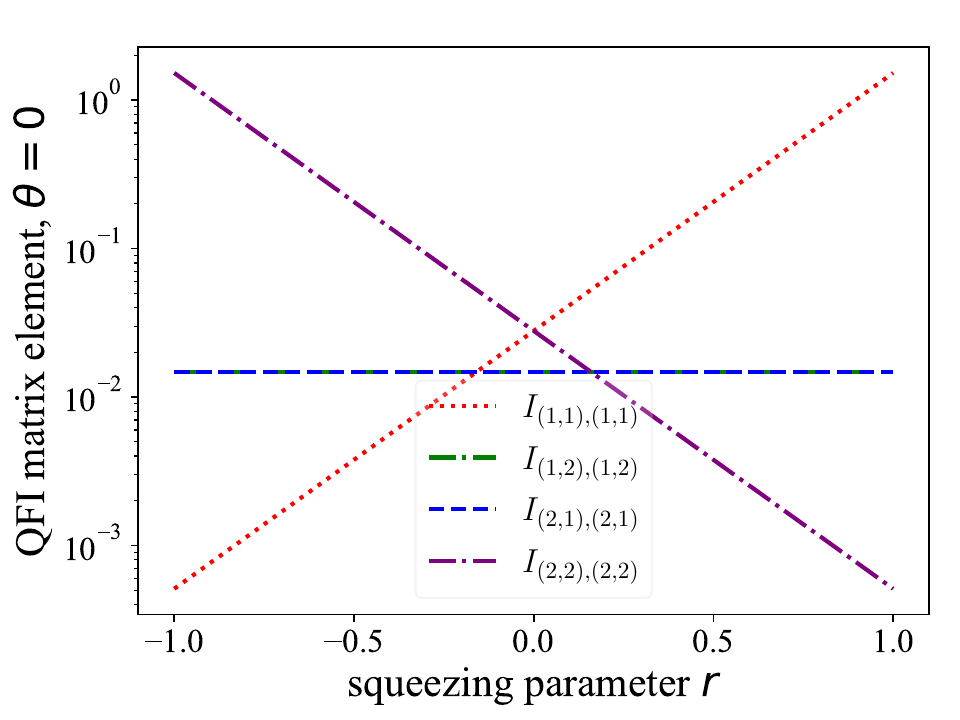}
\includegraphics[width=0.5\textwidth]{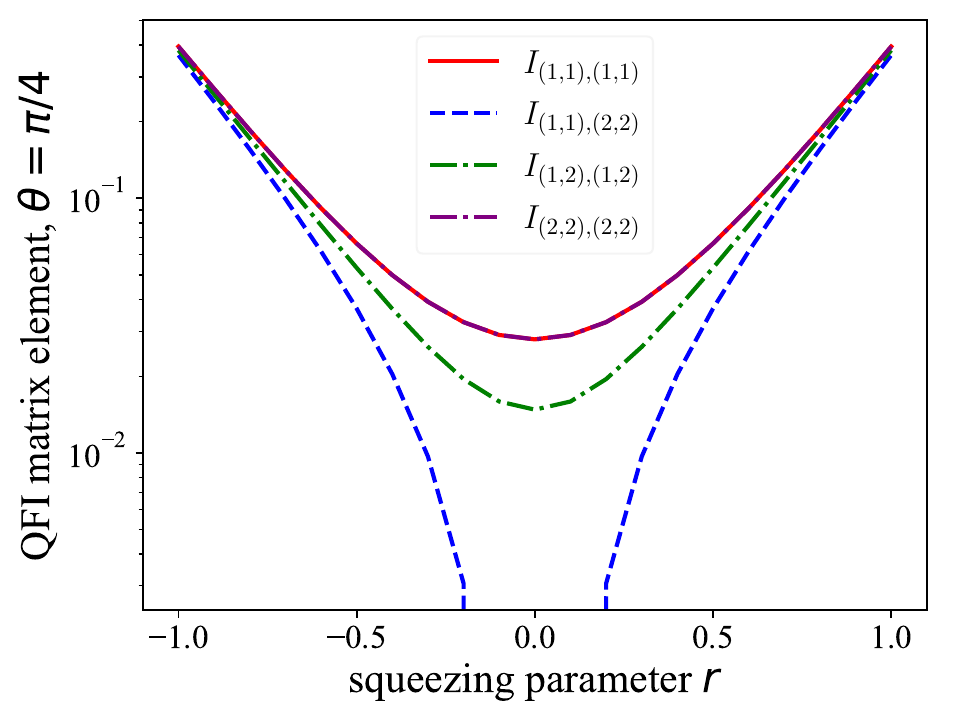}
\caption{Select Fisher--Bures information matrix elements as a function of
the squeezing parameter $r\in [-1,1]$; here, $\epsilon = 0.1$, $\theta = 0$ (top), and $\theta = \pi/4$ (bottom).}
\label{fig:single_mode_sq}
\end{center}
\end{figure}

A physical example is provided by a single-mode optomechanical sensor with Hamiltonian contributions 
$D_1(t) (\hat b + \hat b^\dag)$ and $D_2(\hat b + \hat b^\dag)^2$ \cite{PhysRevA.101.033834}. The first term describes a signal-induced displacement of the mechanical mode, such as that produced by an external force (e.g. gravitational acceleration), while the second describes signal-induced single-mode squeezing, for example through parametric modulation of the mechanical frequency. In this case, our results provide the corresponding information-geometric quantities for estimating the displacement and the quadratic Hamiltonian.

\section{Conclusion}

\label{sec:conclusion}

In conclusion, the main contributions of our paper are formulas for the Fisher--Bures, Kubo--Mori, and $\alpha$-$z$ information matrices of bosonic Gaussian thermal states, as well as formulas for their derivatives and their symmetric logarithmic derivatives. An immediate application of the first contribution, by means of~\eqref{eq:Cramer--Rao-multiple} and~\eqref{eq:scalar-Cramer--Rao-from-matrix-one}, is to provide fundamental limitations on how well one can estimate the mean vector and Hamiltonian matrices of bosonic Gaussian thermal states. Our expressions for their symmetric logarithmic derivatives could be useful in designing optimal strategies for single parameter estimation. Additionally, this finding and the others listed above could be useful if one uses bosonic Gaussian thermal states as an ansatz for quantum machine learning; i.e., the derivative could be useful in developing gradient descent algorithms and the Fisher--Bures, Kubo--Mori, and $\alpha$-$z$ information matrices could be useful in developing natural gradient descent algorithms, similar to the settings considered in~\cite{patel2024quantumboltzmann,patel2024naturalgradient}.

Going forward from here, there is plenty to investigate. Could our findings be helpful in the learning problems recently considered in~\cite{bittel2024optimalestimatestracedistance,fanizza2024efficienthamiltonianstructuretrace}? These papers do not make use of quantum generalizations of Fisher information, but due to their fundamental role in estimation, it seems likely they could play a role in the learning context. It is also pertinent to design a quantum algorithm that can measure the symmetric logarithmic derivative (SLD) observable detailed in Theorem~\ref{thm:SLD-bosonic-G}. As mentioned previously, this will be useful
in single parameter estimation of bosonic Gaussian thermal states. For the finite-dimensional case, an algorithm for measuring the SLD was recently sketched in~\cite{patel2024naturalgradient}. We also wonder whether our notion of the derivative of a bosonic Gaussian thermal state could be useful in the same context in which the alternative notion of derivative from~\cite{bittel2024optimalestimatestracedistance} is. Finally, we mentioned that our formulas for the derivative, the Fisher--Bures, Kubo--Mori, and $\alpha$-$z$ information matrices should have applications in gradient descent and natural gradient descent algorithms when using bosonic Gaussian thermal states as an ansatz. However, we leave the full exploration of these directions to future work.

\begin{acknowledgments}

We thank Ryan Mann for helpful discussions.
ZH is supported by an ARC
DECRA Fellowship (DE230100144) ``Quantum-enabled super-resolution imaging''.
MMW acknowledges support from the National Science Foundation under Grant No.
2304816 and from Air Force Research Laboratory under agreement number FA8750-23-2-0031.

This material is based on research sponsored by Air Force Research Laboratory
under agreement number FA8750-23-2-0031. The U.S. Government is authorized to
reproduce and distribute reprints for Governmental purposes notwithstanding
any copyright notation thereon. The views and conclusions contained herein are
those of the authors and should not be interpreted as necessarily representing
the official policies or endorsements, either expressed or implied, of Air
Force Research Laboratory or the U.S. Government.
\end{acknowledgments}

\section*{Author contributions}

\noindent \textbf{Author Contributions}:
The following describes the different contributions of the authors of this work, using roles defined by the CRediT
(Contributor Roles Taxonomy) project~\cite{NISO}:

\medskip 
\noindent \textbf{ZH}: Formal Analysis, Investigation, Software, Validation, Writing - Original draft,  Writing - Review \& Editing.

\medskip 
\noindent \textbf{MMW}: Conceptualization, Formal Analysis,    Methodology, Validation, Writing - Original draft,  Writing - Review \& Editing.

\bibliography{Ref}

\onecolumngrid

\appendix

\section{QFI matrix parametrized by first and second moments}
\label{sec:comparison}

A useful point of comparison is earlier work deriving the Fisher--Bures, or equivalently quantum Fisher, information for bosonic Gaussian states in a parameterization by first and second moments. In particular, Ref.~\cite{monras2013phasespaceformalismquantum} derived a closed form for the single-parameter case, while Ref.~\cite{Gao2014} presented the corresponding multiparameter extension. Since the notation in these works differs from ours, we briefly summarize the final form in our conventions. This section is adapted from Ref.~\cite{Huang2024}.

Let $\vec{\theta}=(\theta_1,\ldots,\theta_m)$ denote a vector of parameters, let $\vec r$ be the mean vector, and let $\mat \sigma$ be the $2N\times 2N$ covariance matrix. We define
\begin{align}
\partial_j \coloneqq \frac{\partial}{\partial \theta_j}.
\end{align}
We also introduce the vectorized covariance matrix $\vec\varsigma$, obtained by stacking the columns of $\mat\sigma$ into a single $4N^2$-dimensional column vector:
\begin{align}
\vec\varsigma
\coloneqq
(
\sigma_{11},\dotsc,\sigma_{(2N)1},\;
\dotsc,\;
\sigma_{1(2N)},\dotsc,\sigma_{(2N)(2N)}
)^\tp .
\end{align}
Finally, we define the $4N^2\times 4N^2$ matrix
\begin{align}
\mat M
\coloneqq
\mat\sigma\otimes\mat\sigma-\mat\Omega\otimes\mat\Omega,
\end{align}
where $\otimes$ denotes the Kronecker product and $\mat\Omega$ is the symplectic form.

With these definitions, the Fisher--Bures information matrix elements may be written as
\begin{align}
\label{qfi_matrix}
J_{jk}
=
\frac12 (\partial_j \vec\varsigma)^\tp \mat M^{-1} (\partial_k \vec\varsigma)
+
2 (\partial_j \vec r)^\tp \mat\sigma^{-1} (\partial_k \vec r).
\end{align}
This formula makes transparent the two distinct contributions to the information geometry: the first term quantifies the sensitivity of the state through changes in the covariance matrix, i.e., through its squeezing and thermal-noise properties, while the second quantifies the contribution coming from changes in the mean vector, i.e., through displacement. Since $\mat M$ and $\mat\sigma$ are symmetric, the right-hand side is manifestly symmetric under $j\leftrightarrow k$, as required.

To connect with the literature, Ref.~\cite{monras2013phasespaceformalismquantum} considers only a single parameter, so that $\vec\theta$ has one component and $\partial_j$ reduces to an ordinary derivative. In that reference, our $\mat\sigma$, $\mat\Omega$, and $\vec r$ correspond respectively to $\Gamma$, $\omega$, and $d$, while the vectorized covariance matrix is denoted by $\lvert \Gamma)$ (or $(\Gamma\rvert$ in dual form). Ref.~\cite{Gao2014} gives the multiparameter version, but in a different operator basis, namely $(\hat a,\hat a^\dagger)$, in which the covariance matrix is complex symmetric rather than real symmetric, and differs from ours by a factor of two. More precisely, their covariance matrix $\Sigma$ is related to ours by
\begin{align}
\mat\sigma \mapsto 2H\Sigma H^\tp,
\qquad
H=\bigoplus_{k=1}^N \frac{1}{\sqrt2}
\begin{pmatrix}
1&1\\
-i&i
\end{pmatrix},
\end{align}
while the symplectic form transforms as $\mat\Omega \mapsto iH\Omega H^\tp=\Omega$. Consequently,
\begin{align}
\mat M
\mapsto
4(H\otimes H)\,\mathfrak M\,(H^\tp\otimes H^\tp),
\qquad
\mathfrak M=\Sigma\otimes\Sigma+\Omega\otimes\Omega/4,
\end{align}
and similarly $\vec r\mapsto H\lambda$ and $\vec\varsigma\mapsto 2(H\otimes H)\Sigma_{\mathrm{(vec)}}$. Equation~\eqref{qfi_matrix} is therefore just the moment-based Fisher--Bures information matrix expressed in our conventions.

\section{Proof of Proposition~\ref{prop:deriv}}

\label{app:deriv-proof}

Here we follow the proof of \cite[Lemma~10]{patel2024quantumboltzmann} quite closely, with the main change being that $G_j$ gets replaced by $\partial_j G(\theta)$. We provide a detailed proof for completeness.

According to Duhamel's formula, the partial derivative of a matrix exponential
$e^{A(x)}$ with respect to a parameter $x$ is given as follows:
\begin{equation}
\partial_{x}e^{A(x)}=\int_{0}^{1}e^{(1-u)A(x)}\left(  \partial_{x}A(x)\right)
e^{uA(x)}\ du.\label{eq:duh_for}
\end{equation}
Using this formula for $\partial_{j}e^{-G(\theta)}$, we obtain:
\begin{align}
\partial_{j}e^{-G(\theta)} &  =\int_{0}^{1}e^{(1-u)(-G(\theta))}\left(
\partial_{j}(-G(\theta)\right)  )e^{u(-G(\theta))}\ du\\
&  =-\int_{0}^{1}e^{(u-1)G(\theta)}\left(  \partial_{j}G(\theta)\right)
e^{-uG(\theta)}\ du.\label{eq:duh-simplify}
\end{align}
Now, suppose that the spectral decomposition of $G(\theta)$ is as follows:
\begin{equation}
G(\theta)=\sum_{k}\lambda_{k}|k\rangle\!\langle k|,\label{eq:G-spec-decomp}
\end{equation}
where $\{\lambda_{k}\}_{k}$ are the eigenvalues and $\{|k\rangle\}_{k}$ are
the corresponding eigenvectors. Substituting~\eqref{eq:G-spec-decomp}
into~\eqref{eq:duh-simplify}, we find that
\begin{align}
\partial_{j}e^{-G(\theta)} &  =-\int_{0}^{1}\left(  \sum_{k}e^{(u-1)\lambda
_{k}}|k\rangle\!\langle k|\right)  \left(  \partial_{j}G(\theta)\right)
\left(  \sum_{l}e^{-u\lambda_{l}}|l\rangle\!\langle l|\right)  du\\
&  =-\int_{0}^{1}\sum_{k,l}e^{(u-1)\lambda_{k}}|k\rangle\!\langle k|\left(
\partial_{j}G(\theta)\right)  e^{-u\lambda_{l}}|l\rangle\!\langle l|\ du\\
&  =-\sum_{k,l}|k\rangle\!\langle k|\left(  \partial_{j}G(\theta)\right)
|l\rangle\!\langle l|\left(  \int_{0}^{1}e^{(u-1)\lambda_{k}}e^{-u\lambda_{l}
}\ du\right)  \\
&  =-\sum_{k,l}|k\rangle\!\langle k|\left(  \partial_{j}G(\theta)\right)
|l\rangle\!\langle l|\left(  e^{-\lambda_{k}}\int_{0}^{1}e^{u\left(
\lambda_{k}-\lambda_{l}\right)  }\ du\right)  \\
&  =-\sum_{k,l}|k\rangle\!\langle k|\left(  \partial_{j}G(\theta)\right)
|l\rangle\!\langle l|\left(  e^{-\lambda_{k}}\frac{e^{\lambda_{k}-\lambda_{l}
}-1}{\lambda_{k}-\lambda_{l}}\right)  .\label{eq:pd-duh-sd}
\end{align}
Now, consider the following:
\begin{equation}
e^{-\lambda_{k}}\frac{e^{\lambda_{k}-\lambda_{l}}-1}{\lambda_{k}-\lambda_{l}
}=e^{-\lambda_{k}}\frac{e^{\lambda_{k}-\lambda_{l}}-1}{e^{\lambda_{k}
-\lambda_{l}}+1}\frac{e^{\lambda_{k}-\lambda_{l}}+1}{\lambda_{k}-\lambda_{l}
}=\frac{\tanh\!{\left(  \frac{\lambda_{k}-\lambda_{l}}{2}\right)  }}
{\frac{\lambda_{k}-\lambda_{l}}{2}}\frac{e^{-\lambda_{l}}+e^{-\lambda_{k}}}
{2}.\label{eq:tanh-intro}
\end{equation}
Let $p(t)$ be a function such that its Fourier transform is the following:
\begin{equation}
\int_{\mathbb{R}}dt\ p(t)e^{-i\omega t}=\frac{\tanh{\frac{\omega}{2}}}
{\frac{\omega}{2}}.\label{eq:pt-fourier}
\end{equation}
Using this equation and~\eqref{eq:tanh-intro}, we rewrite~\eqref{eq:pd-duh-sd}
in the following way:
\begin{align}
&  \partial_{j}e^{-G(\theta)}\nonumber\\
&  =-\sum_{k,l}|k\rangle\!\langle k|\left(  \partial_{j}G(\theta)\right)
|l\rangle\!\langle l|\left(  \left(  \int_{\mathbb{R}}dt\ p(t)e^{-i\left(
\lambda_{k}-\lambda_{l}\right)  t}\right)  \frac{e^{-\lambda_{l}}
+e^{-\lambda_{k}}}{2}\right)  \\
&  =-\frac{1}{2}\sum_{k,l}|k\rangle\!\langle k|\left(  \partial_{j}
G(\theta)\right)  |l\rangle\!\langle l|\left(  \int_{\mathbb{R}}
dt\ p(t)\left(  e^{-i\lambda_{k}t+i\lambda_{l}t-\lambda_{l}}+e^{-i\lambda
_{k}t-\lambda_{k}+i\lambda_{l}t}\right)  \right)  \\
&  =-\frac{1}{2}\left(  \int_{\mathbb{R}}dt\ p(t)\left(  \sum_{k,l}
|k\rangle\!\langle k|\partial_{j}G(\theta)|l\rangle\!\langle l|\ e^{-i\lambda
_{k}t+i\lambda_{l}t-\lambda_{l}}+\sum_{k,l}|k\rangle\!\langle k|\left(
\partial_{j}G(\theta)\right)  |l\rangle\!\langle l|\ e^{-i\lambda_{k}
t-\lambda_{k}+i\lambda_{l}t}\right)  \right)  \\
&  =-\frac{1}{2}\Bigg(\int_{\mathbb{R}}dt\ p(t)\Bigg(\sum_{k}e^{-i\lambda
_{k}t}|k\rangle\!\langle k|\left(  \partial_{j}G(\theta)\right)  \sum
_{l}e^{i\lambda_{l}t-\lambda_{l}}|l\rangle\!\langle l|\nonumber\\
&  \hspace{8cm}+\sum_{k}e^{-i\lambda_{k}t-\lambda_{k}}|k\rangle\!\langle
k|\left(  \partial_{j}G(\theta)\right)  \sum_{l}e^{i\lambda_{l}t}
|l\rangle\!\langle l|\Bigg)\Bigg)\\
&  =-\frac{1}{2}\left(  \int_{\mathbb{R}}dt\ p(t)\left(  e^{-iG(\theta
)t}\left(  \partial_{j}G(\theta)\right)  e^{iG(\theta)t}e^{-G(\theta
)}+e^{-G(\theta)}e^{-iG(\theta)t}\left(  \partial_{j}G(\theta)\right)
e^{iG(\theta)t}\right)  \right)  \\
&  =-\frac{1}{2}\left(  \left(  \int_{\mathbb{R}}dt\ p(t)e^{-iG(\theta
)t}\left(  \partial_{j}G(\theta)\right)  e^{iG(\theta)t}\right)
e^{-G(\theta)}+e^{-G(\theta)}\left(  \int_{\mathbb{R}}dt\ p(t)e^{-iG(\theta
)t}\left(  \partial_{j}G(\theta)\right)  e^{iG(\theta)t}\right)  \right)  \\
&  =-\frac{1}{2}\left(  \left(  \int_{\mathbb{R}}dt\ p(t)e^{iG(\theta
)t}\left(  \partial_{j}G(\theta)\right)  e^{-iG(\theta)t}\right)
e^{-G(\theta)}+e^{-G(\theta)}\left(  \int_{\mathbb{R}}dt\ p(t)e^{iG(\theta
)t}\left(  \partial_{j}G(\theta)\right)  e^{-iG(\theta)t}\right)  \right)  \label{eq:pt-even} \\
&  =-\frac{1}{2}\left(  \Phi_{\theta}(\partial_{j}G(\theta))e^{-G(\theta
)}+e^{-G(\theta)}\Phi_{\theta}(\partial_{j}G(\theta))\right)  \\
&  =-\frac{1}{2}\left\{  \Phi_{\theta}(\partial_{j}G(\theta)),e^{-G(\theta
)}\right\}  ,
\end{align}
where, in the penultimate equality, we use the definition of the quantum
channel $\Phi_{\theta}$ from~\eqref{eq:Phi-channel}. For \eqref{eq:pt-even}, we also used the fact that $p(t)$ is an even function.

\section{Proof of Proposition~\ref{prop:FB-KM-info-mats}}

\label{app:FB-KM-proofs}

\subsection{Proof of Equation~\eqref{eq:FB-info-mat-gen-thermal}}

Here we follow the proof of \cite[Theorem~1]{patel2024naturalgradient} quite closely, with the main change being that $G_i$ gets replaced by $\partial_i G(\theta)$. We provide a detailed proof for completeness.

Plugging~\eqref{eq:thermal-state-deriv} into~\eqref{eq:qfim-explicit}, we find
that we should consider the following term:
\begin{align}
\langle k|\partial_{i}\rho(\theta)|\ell\rangle &  =\langle k|\left[  -\frac
{1}{2}\left\{  \Phi_{\theta}(\left(  \partial_{i}G(\theta)\right)
),\rho(\theta)\right\}  +\rho(\theta)\left\langle \partial_{i}G(\theta
)\right\rangle _{\rho(\theta)}\right]  |\ell\rangle\\
&  =-\frac{1}{2}\langle k|\left\{  \Phi_{\theta}(\partial_{i}G(\theta
)),\rho(\theta)\right\}  |\ell\rangle+\langle k|\rho(\theta)|\ell
\rangle\left\langle \partial_{i}G(\theta)\right\rangle _{\rho(\theta)}\\
&  =-\frac{1}{2}\left(  \langle k|\Phi_{\theta}(\partial_{i}G(\theta
))|\ell\rangle\lambda_{\ell}+\langle k|\Phi_{\theta}(\partial_{i}
G(\theta))|\ell\rangle\lambda_{k}\right)  +\delta_{k\ell}\lambda_{\ell
}\left\langle \partial_{i}G(\theta)\right\rangle _{\rho(\theta)}\\
&  =-\frac{1}{2}\langle k|\Phi_{\theta}(\partial_{i}G(\theta))|\ell
\rangle\left(  \lambda_{k}+\lambda_{\ell}\right)  +\delta_{k\ell}\lambda
_{\ell}\left\langle \partial_{i}G(\theta)\right\rangle _{\rho(\theta
)}.\label{eq:basic-calc-deriv}
\end{align}
This implies that
\begin{equation}
\langle\ell|\partial_{j}\rho(\theta)|k\rangle=-\frac{1}{2}\langle\ell
|\Phi_{\theta}(\partial_{j}G(\theta))|k\rangle\left(  \lambda_{k}
+\lambda_{\ell}\right)  +\delta_{k\ell}\lambda_{\ell}\left\langle \partial
_{j}G(\theta)\right\rangle _{\rho(\theta)}\label{eq:basic-calc-deriv-HC}.
\end{equation}
Plugging~\eqref{eq:basic-calc-deriv} and~\eqref{eq:basic-calc-deriv-HC}
into~\eqref{eq:qfim-explicit}, we find that
\begin{align}
I_{ij}^{F} &  =2\sum_{k,\ell}\frac{\langle k|\partial_{i}\rho(\theta
)|\ell\rangle\!\langle\ell|\partial_{j}\rho(\theta)|k\rangle}{\lambda
_{k}+\lambda_{\ell}}\nonumber\\
&  =2\sum_{k,\ell}\frac{\left[
\begin{array}
[c]{c}
\left(  -\frac{1}{2}\langle k|\Phi_{\theta}(\partial_{i}G(\theta))|\ell
\rangle\left(  \lambda_{k}+\lambda_{\ell}\right)  +\delta_{k\ell}\lambda
_{\ell}\left\langle \partial_{i}G(\theta)\right\rangle _{\rho(\theta)}\right)
\times\\
\left(  -\frac{1}{2}\langle\ell|\Phi_{\theta}(\partial_{j}G(\theta
))|k\rangle\left(  \lambda_{k}+\lambda_{\ell}\right)  +\delta_{k\ell}
\lambda_{\ell}\left\langle \partial_{j}G(\theta)\right\rangle _{\rho(\theta
)}\right)
\end{array}
\right]  }{\lambda_{k}+\lambda_{\ell}}\\
&  =2\sum_{k,\ell}\frac{1}{4}\frac{\langle k|\Phi_{\theta}(\partial
_{i}G(\theta))|\ell\rangle\!\langle\ell|\Phi_{\theta}(\partial_{j}
G(\theta))|k\rangle\left(  \lambda_{k}+\lambda_{\ell}\right)  ^{2}}
{\lambda_{k}+\lambda_{\ell}}\nonumber\\
&  \qquad+2\sum_{k,\ell}\left(  -\frac{1}{2}\right)  \frac{\langle
k|\Phi_{\theta}(\partial_{i}G(\theta))|\ell\rangle\left(  \lambda_{k}
+\lambda_{\ell}\right)  \delta_{k\ell}\lambda_{\ell}\left\langle \partial
_{j}G(\theta)\right\rangle _{\rho(\theta)}}{\lambda_{k}+\lambda_{\ell}
}\nonumber\\
&  \qquad+2\sum_{k,\ell}\left(  -\frac{1}{2}\right)  \frac{\langle\ell
|\Phi_{\theta}(\partial_{j}G(\theta))|k\rangle\left(  \lambda_{k}
+\lambda_{\ell}\right)  \delta_{k\ell}\lambda_{\ell}\left\langle \partial
_{i}G(\theta)\right\rangle _{\rho(\theta)}}{\lambda_{k}+\lambda_{\ell}
}\nonumber\\
&  \qquad+2\sum_{k,\ell}\frac{\delta_{k\ell}\lambda_{\ell}\left\langle
\partial_{i}G(\theta)\right\rangle _{\rho(\theta)}\delta_{k\ell}\lambda_{\ell
}\left\langle \partial_{j}G(\theta)\right\rangle _{\rho(\theta)}}{\lambda
_{k}+\lambda_{\ell}}\\
&  =\frac{1}{2}\sum_{k,\ell}\langle k|\Phi_{\theta}(\partial_{i}
G(\theta))|\ell\rangle\langle\ell|\Phi_{\theta}(\partial_{j}G(\theta
))|k\rangle\left(  \lambda_{k}+\lambda_{\ell}\right)  \nonumber\\
&  \qquad-\sum_{k,\ell}\langle k|\Phi_{\theta}(\partial_{i}G(\theta
))|\ell\rangle\delta_{k\ell}\lambda_{\ell}\left\langle \partial_{j}
G(\theta)\right\rangle _{\rho(\theta)}-\sum_{k,\ell}\langle\ell|\Phi_{\theta
}(\partial_{j}G(\theta))|k\rangle\delta_{k\ell}\lambda_{\ell}\left\langle
\partial_{i}G(\theta)\right\rangle _{\rho(\theta)}\nonumber\\
&  \qquad+2\sum_{k,\ell}\frac{\delta_{k\ell}\lambda_{\ell}\left\langle
\partial_{i}G(\theta)\right\rangle _{\rho(\theta)}\delta_{k\ell}\lambda_{\ell
}\left\langle \partial_{j}G(\theta)\right\rangle _{\rho(\theta)}}{\lambda
_{k}+\lambda_{\ell}}\\
&  =\frac{1}{2}\sum_{k,\ell}\langle k|\Phi_{\theta}(\partial_{i}
G(\theta))|\ell\rangle\langle\ell|\Phi_{\theta}(\partial_{j}G(\theta
))|k\rangle\lambda_{k}+\frac{1}{2}\sum_{k,\ell}\langle k|\Phi_{\theta
}(\partial_{i}G(\theta))|\ell\rangle\!\langle\ell|\Phi_{\theta}(\partial
_{j}G(\theta))|k\rangle\lambda_{\ell}\nonumber\\
&  \qquad-\sum_{k}\langle k|\Phi_{\theta}(\partial_{i}G(\theta))|k\rangle
\lambda_{k}\left\langle \partial_{j}G(\theta)\right\rangle _{\rho(\theta
)}-\sum_{k}\langle k|\Phi_{\theta}(\partial_{j}G(\theta))|k\rangle\lambda
_{k}\left\langle \partial_{i}G(\theta)\right\rangle _{\rho(\theta)}\nonumber\\
&  \qquad+2\sum_{k}\frac{\lambda_{k}^{2}\left\langle \partial_{i}
G(\theta)\right\rangle _{\rho(\theta)}\left\langle \partial_{j}G(\theta
)\right\rangle _{\rho(\theta)}}{2\lambda_{k}}\\
&  =\frac{1}{2}\operatorname{Tr}[\Phi_{\theta}(\partial_{j}G(\theta
))\rho(\theta)\Phi_{\theta}(G_{i})]+\frac{1}{2}\operatorname{Tr}[\Phi_{\theta
}(\partial_{i}G(\theta))\rho(\theta)\Phi_{\theta}(\partial_{j}G(\theta
))]\nonumber\\
&  \qquad-\operatorname{Tr}[\rho(\theta)\Phi_{\theta}(\partial_{i}
G(\theta))]\left\langle \partial_{j}G(\theta)\right\rangle _{\rho(\theta
)}-\operatorname{Tr}[\rho(\theta)\Phi_{\theta}(\partial_{j}G(\theta
))]\left\langle \partial_{i}G(\theta)\right\rangle _{\rho(\theta)}\nonumber\\
&  \qquad+\left\langle \partial_{i}G(\theta)\right\rangle _{\rho(\theta
)}\left\langle \partial_{j}G(\theta)\right\rangle _{\rho(\theta)}\\
&  =\frac{1}{2}\operatorname{Tr}[\Phi_{\theta}(\partial_{j}G(\theta
))\rho(\theta)\Phi_{\theta}(\partial_{i}G(\theta))]+\frac{1}{2}
\operatorname{Tr}[\Phi_{\theta}(\partial_{i}G(\theta))\rho(\theta)\Phi
_{\theta}(\partial_{j}G(\theta))]\nonumber\\
&  \qquad-\left\langle \partial_{i}G(\theta)\right\rangle _{\rho(\theta
)}\left\langle \partial_{j}G(\theta)\right\rangle _{\rho(\theta)}\\
&  =\frac{1}{2}\left\langle \left\{  \Phi_{\theta}(\partial_{i}G(\theta
)),\Phi_{\theta}(\partial_{j}G(\theta))\right\}  \right\rangle _{\rho(\theta
)}-\left\langle \partial_{i}G(\theta)\right\rangle _{\rho(\theta)}\left\langle
\partial_{j}G(\theta)\right\rangle _{\rho(\theta)}.
\end{align}

\subsection{Proof of Equation~\eqref{eq:KM-info-mat-gen-thermal}}

Here we follow the proof of \cite[Theorem~2]{patel2024naturalgradient} quite closely, with the main change being that $G_i$ gets replaced by $\partial_i G(\theta)$. We provide a detailed proof for completeness.

Consider that
\begin{align}
I_{ij}^{\operatorname{KM}}(\theta) &  =\sum_{k,\ell}\frac{\ln\lambda_{k}
-\ln\lambda_{\ell}}{\lambda_{k}-\lambda_{\ell}}\langle k|\partial_{i}
\rho(\theta)|\ell\rangle\!\langle\ell|\partial_{j}\rho(\theta)|k\rangle\\
&  =\sum_{k,\ell}\frac{\ln\lambda_{k}-\ln\lambda_{\ell}}{\lambda_{k}
-\lambda_{\ell}}\left[
\begin{array}
[c]{c}
\left(  -\frac{1}{2}\langle k|\Phi_{\theta}(\partial_{i}G(\theta))|\ell
\rangle\left(  \lambda_{k}+\lambda_{\ell}\right)  +\delta_{k\ell}\lambda
_{\ell}\left\langle \partial_{i}G(\theta)\right\rangle _{\rho(\theta)}\right)
\times\\
\left(  -\frac{1}{2}\langle\ell|\Phi_{\theta}(\partial_{j}G(\theta
))|k\rangle\left(  \lambda_{k}+\lambda_{\ell}\right)  +\delta_{k\ell}
\lambda_{\ell}\left\langle \partial_{j}G(\theta)\right\rangle _{\rho(\theta
)}\right)
\end{array}
\right]  \\
&  =\sum_{k,\ell}\frac{1}{4}\frac{\ln\lambda_{k}-\ln\lambda_{\ell}}
{\lambda_{k}-\lambda_{\ell}}\langle k|\Phi_{\theta}(\partial_{i}
G(\theta))|\ell\rangle\!\langle\ell|\Phi_{\theta}(\partial_{j}G(\theta
))|k\rangle\left(  \lambda_{k}+\lambda_{\ell}\right)  ^{2}\nonumber\\
&  \qquad+\sum_{k,\ell}-\frac{1}{2}\frac{\ln\lambda_{k}-\ln\lambda_{\ell}
}{\lambda_{k}-\lambda_{\ell}}\langle k|\Phi_{\theta}(\partial_{i}
G(\theta))|\ell\rangle\left(  \lambda_{k}+\lambda_{\ell}\right)  \delta
_{k\ell}\lambda_{\ell}\left\langle \partial_{j}G(\theta)\right\rangle
_{\rho(\theta)}\nonumber\\
&  \qquad+\sum_{k,\ell}-\frac{1}{2}\frac{\ln\lambda_{k}-\ln\lambda_{\ell}
}{\lambda_{k}-\lambda_{\ell}}\langle\ell|\Phi_{\theta}(\partial_{j}
G(\theta))|k\rangle\left(  \lambda_{k}+\lambda_{\ell}\right)  \delta_{k\ell
}\lambda_{\ell}\left\langle \partial_{i}G(\theta)\right\rangle _{\rho(\theta
)}\nonumber\\
&  \qquad+\sum_{k,\ell}\frac{\ln\lambda_{k}-\ln\lambda_{\ell}}{\lambda
_{k}-\lambda_{\ell}}\delta_{k\ell}\lambda_{\ell}\left\langle \partial
_{i}G(\theta)\right\rangle _{\rho(\theta)}\delta_{k\ell}\lambda_{\ell
}\left\langle \partial_{j}G(\theta)\right\rangle _{\rho(\theta)}\\
&  =\sum_{k,\ell}\frac{1}{4}\frac{(\ln\lambda_{k}-\ln\lambda_{\ell})\left(
\lambda_{k}+\lambda_{\ell}\right)  ^{2}}{\lambda_{k}-\lambda_{\ell}}\langle
k|\Phi_{\theta}(\partial_{i}G(\theta))|\ell\rangle\!\langle\ell|\Phi_{\theta
}(\partial_{j}G(\theta))|k\rangle\nonumber\\
&  \qquad+\sum_{k}\left(  -\frac{1}{2}\right)  \frac{1}{\lambda_{k}}\langle
k|\Phi_{\theta}(\partial_{i}G(\theta))|k\rangle\left(  2\lambda_{k}\right)
\lambda_{k}\left\langle \partial_{j}G(\theta)\right\rangle _{\rho(\theta
)}\nonumber\\
&  \qquad+\sum_{k}\left(  -\frac{1}{2}\right)  \frac{1}{\lambda_{k}}\langle
k|\Phi_{\theta}(\partial_{j}G(\theta))|k\rangle\left(  2\lambda_{k}\right)
\lambda_{k}\left\langle \partial_{i}G(\theta)\right\rangle _{\rho(\theta
)}\nonumber\\
&  \qquad+\sum_{k}\frac{1}{\lambda_{k}}\lambda_{k}\left\langle \partial
_{i}G(\theta)\right\rangle _{\rho(\theta)}\lambda_{k}\left\langle \partial
_{j}G(\theta)\right\rangle _{\rho(\theta)}\\
&  =\sum_{k,\ell}\frac{1}{4}\frac{(\ln\lambda_{k}-\ln\lambda_{\ell})\left(
\lambda_{k}+\lambda_{\ell}\right)  ^{2}}{\lambda_{k}-\lambda_{\ell}}\langle
k|\Phi_{\theta}(\partial_{i}G(\theta))|\ell\rangle\!\langle\ell|\Phi_{\theta
}(\partial_{j}G(\theta))|k\rangle\nonumber\\
&  \qquad-\sum_{k}\langle k|\Phi_{\theta}(\partial_{i}G(\theta))|k\rangle
\lambda_{k}\left\langle \partial_{j}G(\theta)\right\rangle _{\rho(\theta
)}-\sum_{k}\langle k|\Phi_{\theta}(\partial_{j}G(\theta))|k\rangle\lambda
_{k}\left\langle \partial_{i}G(\theta)\right\rangle _{\rho(\theta)}\\
&  \qquad+\sum_{k}\left\langle \partial_{i}G(\theta)\right\rangle
_{\rho(\theta)}\lambda_{k}\left\langle \partial_{j}G(\theta)\right\rangle
_{\rho(\theta)}\\
&  =\sum_{k,\ell}\frac{1}{4}\frac{(\ln\lambda_{k}-\ln\lambda_{\ell})\left(
\lambda_{k}+\lambda_{\ell}\right)  ^{2}}{\lambda_{k}-\lambda_{\ell}}\langle
k|\Phi_{\theta}(\partial_{i}G(\theta))|\ell\rangle\!\langle\ell|\Phi_{\theta
}(\partial_{j}G(\theta))|k\rangle\nonumber\\
&  \qquad-\operatorname{Tr}\left[  \rho(\theta)\Phi_{\theta}(\partial
_{i}G(\theta))\right]  \left\langle \partial_{j}G(\theta)\right\rangle
_{\rho(\theta)}-\operatorname{Tr}\left[  \rho(\theta)\Phi_{\theta}
(\partial_{j}G(\theta))\right]  \left\langle \partial_{i}G(\theta
)\right\rangle _{\rho(\theta)}+\left\langle \partial_{i}G(\theta)\right\rangle
_{\rho(\theta)}\left\langle \partial_{j}G(\theta)\right\rangle _{\rho(\theta
)}\\
&  =\sum_{k,\ell}\frac{1}{4}\frac{(\ln\lambda_{k}-\ln\lambda_{\ell})\left(
\lambda_{k}+\lambda_{\ell}\right)  ^{2}}{\lambda_{k}-\lambda_{\ell}}\langle
k|\Phi_{\theta}(\partial_{i}G(\theta))|\ell\rangle\!\langle\ell|\Phi_{\theta
}(\partial_{j}G(\theta))|k\rangle-\left\langle \partial_{i}G(\theta
)\right\rangle _{\rho(\theta)}\left\langle \partial_{j}G(\theta)\right\rangle
_{\rho(\theta)}.\label{eq:KM-expansion}
\end{align}
The fourth equality is a consequence of the following fact:
\begin{equation}
\lim_{x\rightarrow y}\frac{\ln x-\ln y}{x-y}=\frac{1}{y}.
\end{equation}
Now suppose that $G(\theta)=\sum_{k}\mu_{k}|k\rangle\!\langle k|$. This
implies that for all $k$, we have $\lambda_{k}=\frac{e^{-\mu_{k}}}{Z}$, where
$Z$ is the partition function. Plugging this into~\eqref{eq:KM-expansion}, we
obtain:
\begin{align}
&  \sum_{k,\ell}\frac{1}{4}\frac{(\ln\lambda_{k}-\ln\lambda_{\ell})\left(
\lambda_{k}+\lambda_{\ell}\right)  ^{2}}{\lambda_{k}-\lambda_{\ell}}\langle
k|\Phi_{\theta}(\partial_{i}G(\theta))|\ell\rangle\!\langle\ell|\Phi_{\theta
}(\partial_{j}G(\theta))|k\rangle\nonumber\\
&  =\sum_{k,\ell}\frac{1}{4}\frac{\left(  \ln\frac{e^{-\mu_{k}}}{Z}-\ln
\frac{e^{-\mu_{\ell}}}{Z}\right)  \left(  \frac{e^{-\mu_{k}}}{Z}+\frac
{e^{-\mu_{\ell}}}{Z}\right)  ^{2}}{\frac{e^{-\mu_{k}}}{Z}-\frac{e^{-\mu_{\ell
}}}{Z}}\langle k|\Phi_{\theta}(\partial_{i}G(\theta))|\ell\rangle\!\langle
\ell|\Phi_{\theta}(\partial_{j}G(\theta))|k\rangle\\
&  =\sum_{k,\ell}\frac{1}{4Z}\frac{\left(  -\mu_{k}+\mu_{\ell}\right)  \left(
e^{-\mu_{k}}+e^{-\mu_{\ell}}\right)  ^{2}}{e^{-\mu_{k}}-e^{-\mu_{\ell}}
}\langle k|\Phi_{\theta}(\partial_{i}G(\theta))|\ell\rangle\!\langle\ell
|\Phi_{\theta}(\partial_{j}G(\theta))|k\rangle\\
&  =\sum_{k,\ell}\frac{1}{4Z}\frac{\left(  -\mu_{k}+\mu_{\ell}\right)  \left(
e^{-\mu_{k}}+e^{-\mu_{\ell}}\right)  }{\frac{e^{-\mu_{k}}-e^{-\mu_{\ell}}
}{e^{-\mu_{k}}+e^{-\mu_{\ell}}}}\langle k|\Phi_{\theta}(\partial_{i}
G(\theta))|\ell\rangle\!\langle\ell|\Phi_{\theta}(\partial_{j}G(\theta
))|k\rangle\\
&  =\sum_{k,\ell}\frac{1}{4Z}\frac{\left(  -\mu_{k}+\mu_{\ell}\right)  \left(
e^{-\mu_{k}}+e^{-\mu_{\ell}}\right)  }{\frac{e^{-\mu_{k}+\mu_{\ell}}
-1}{e^{-\mu_{k}+\mu_{\ell}}+1}}\langle k|\Phi_{\theta}(\partial_{i}
G(\theta))|\ell\rangle\!\langle\ell|\Phi_{\theta}(\partial_{j}G(\theta
))|k\rangle\\
&  =\sum_{k,\ell}\frac{1}{4Z}\frac{\left(  -\mu_{k}+\mu_{\ell}\right)  \left(
e^{-\mu_{k}}+e^{-\mu_{\ell}}\right)  }{\tanh\!\left(  \frac{-\mu_{k}+\mu
_{\ell}}{2}\right)  }\langle k|\Phi_{\theta}(\partial_{i}G(\theta
))|\ell\rangle\!\langle\ell|\Phi_{\theta}(\partial_{j}G(\theta))|k\rangle\\
&  =\sum_{k,\ell}\frac{1}{4Z}\frac{\left(  -\mu_{k}+\mu_{\ell}\right)  \left(
e^{-\mu_{k}}+e^{-\mu_{\ell}}\right)  }{\tanh\!\left(  \frac{-\mu_{k}+\mu
_{\ell}}{2}\right)  }\langle k|\int_{\mathbb{R}}dt\ p(t)\ e^{iG(\theta
)t}\partial_{i}G(\theta)e^{-iG(\theta)t}|\ell\rangle\!\langle\ell|\Phi
_{\theta}(\partial_{j}G(\theta))|k\rangle\\
&  =\sum_{k,\ell}\frac{1}{4Z}\frac{\left(  -\mu_{k}+\mu_{\ell}\right)  \left(
e^{-\mu_{k}}+e^{-\mu_{\ell}}\right)  }{\tanh\!\left(  \frac{-\mu_{k}+\mu
_{\ell}}{2}\right)  }\times\nonumber\\
&  \qquad\langle k|\int_{\mathbb{R}}dt\ p(t)\ \left(  \sum_{m}|m\rangle
\!\langle m|e^{i\mu_{m}t}\right)  \partial_{i}G(\theta)\left(  \sum
_{n}|n\rangle\!\langle n|e^{-i\mu_{n}t}\right)  |\ell\rangle\!\langle\ell
|\Phi_{\theta}(\partial_{j}G(\theta))|k\rangle\\
&  =\sum_{k,\ell}\frac{1}{4Z}\frac{\left(  -\mu_{k}+\mu_{\ell}\right)  \left(
e^{-\mu_{k}}+e^{-\mu_{\ell}}\right)  }{\tanh\!\left(  \frac{-\mu_{k}+\mu
_{\ell}}{2}\right)  }\int_{\mathbb{R}}dt\ p(t)\ e^{i\mu_{k}t}\langle
k|\partial_{i}G(\theta)|\ell\rangle e^{-i\mu_{\ell}t}\langle\ell|\Phi_{\theta
}(\partial_{j}G(\theta))|k\rangle\\
&  =\sum_{k,\ell}\frac{1}{4Z}\frac{\left(  -\mu_{k}+\mu_{\ell}\right)  \left(
e^{-\mu_{k}}+e^{-\mu_{\ell}}\right)  }{\tanh\!\left(  \frac{-\mu_{k}+\mu
_{\ell}}{2}\right)  }\int_{\mathbb{R}}dt\ p(t)\ e^{-i(\mu_{\ell}-\mu_{k}
)t}\langle k|\partial_{i}G(\theta)|\ell\rangle\langle\ell|\Phi_{\theta
}(\partial_{j}G(\theta))|k\rangle\\
&  =\sum_{k,\ell}\frac{1}{4Z}\frac{\left(  -\mu_{k}+\mu_{\ell}\right)  \left(
e^{-\mu_{k}}+e^{-\mu_{\ell}}\right)  }{\tanh\!\left(  \frac{-\mu_{k}+\mu
_{\ell}}{2}\right)  }\frac{\tanh\!\left(  \frac{-\mu_{k}+\mu_{\ell}}
{2}\right)  }{\frac{\left(  -\mu_{k}+\mu_{\ell}\right)  }{2}}\langle
k|\partial_{i}G(\theta)|\ell\rangle\langle\ell|\Phi_{\theta}(\partial
_{j}G(\theta))|k\rangle\\
&  =\sum_{k,\ell}\frac{1}{2Z}\left(  e^{-\mu_{k}}+e^{-\mu_{\ell}}\right)
\langle k|\partial_{i}G(\theta)|\ell\rangle\langle\ell|\Phi_{\theta}
(\partial_{j}G(\theta))|k\rangle\\
&  =\sum_{k,\ell}\frac{1}{2}\left(  \frac{e^{-\mu_{k}}}{Z}+\frac{e^{-\mu
_{\ell}}}{Z}\right)  \langle k|\partial_{i}G(\theta)|\ell\rangle\langle
\ell|\Phi_{\theta}(\partial_{j}G(\theta))|k\rangle\\
&  =\sum_{k,\ell}\frac{1}{2}\left(  \lambda_{k}+\lambda_{\ell}\right)  \langle
k|\partial_{i}G(\theta)|\ell\rangle\langle\ell|\Phi_{\theta}(\partial
_{j}G(\theta))|k\rangle\\
&  =\frac{1}{2}\operatorname{Tr}\left[  \rho(\theta)\partial_{i}G(\theta
)\Phi_{\theta}(\partial_{j}G(\theta))\right]  +\frac{1}{2}\operatorname{Tr}
\left[  G_{i}\rho(\theta)\Phi_{\theta}(\partial_{j}G(\theta))\right]  \\
&  =\frac{1}{2}\operatorname{Tr}\left[  \{\partial_{i}G(\theta),\Phi_{\theta
}(\partial_{j}G(\theta))\}\rho(\theta)\right]  .\label{eq:last-line-KM}
\end{align}
When combining~\eqref{eq:last-line-KM} with~\eqref{eq:KM-expansion}, the proof
is concluded. Let us note that we made use of the Fourier transform relation reviewed in \cite[Lemma~12]{patel2024quantumboltzmann}.

\section{Proof of Proposition~\ref{prop:FB-info-mats-alt}}

\label{app:alt-form-FB-inf-mat}Consider that
\begin{align}
&  \left\langle \left\{  \Phi_{\theta}(\partial_{i}G(\theta)),\Phi_{\theta
}(\partial_{j}G(\theta))\right\}  \right\rangle _{\rho(\theta)}\nonumber\\
&  =\int_{\mathbb{R}}\int_{\mathbb{R}}dt_{1}\ dt_{2}\ p(t_{1})\ p(t_{2}
)\ \left\langle \left\{  e^{iG(\theta)t_{1}}(\partial_{i}G(\theta
))e^{-iG(\theta)t_{1}},e^{iG(\theta)t_{2}}(\partial_{j}G(\theta))e^{-iG(\theta
)t_{2}}\right\}  \right\rangle _{\rho(\theta)}\\
&  =\int_{\mathbb{R}}\int_{\mathbb{R}}dt_{1}\ dt_{2}\ p(t_{1})\ p(t_{2}
)\ \operatorname{Tr}\left[  \left\{  e^{iG(\theta)t_{1}}(\partial_{i}
G(\theta))e^{-iG(\theta)t_{1}},e^{iG(\theta)t_{2}}(\partial_{j}G(\theta
))e^{-iG(\theta)t_{2}}\right\}  \rho(\theta)\right]  \\
&  =\int_{\mathbb{R}}\int_{\mathbb{R}}dt_{1}\ dt_{2}\ p(t_{1})\ p(t_{2}
)\ \operatorname{Tr}\left[  \left\{  (\partial_{i}G(\theta)),e^{iG(\theta
)\left(  t_{2}-t_{1}\right)  }(\partial_{j}G(\theta))e^{-iG(\theta)\left(
t_{2}-t_{1}\right)  }\right\}  \rho(\theta)\right]  \\
&  =\int_{\mathbb{R}}\int_{\mathbb{R}}d\tau\ dt\ p(\tau)\ p(t+\tau
)\ \operatorname{Tr}\left[  \left\{  (\partial_{i}G(\theta)),e^{iG(\theta
)t}(\partial_{j}G(\theta))e^{-iG(\theta)t}\right\}  \rho(\theta)\right]  \\
&  =\int_{\mathbb{R}}\ dt\ q(t)\ \operatorname{Tr}\left[  \left\{
(\partial_{i}G(\theta)),e^{iG(\theta)t}(\partial_{j}G(\theta))e^{-iG(\theta
)t}\right\}  \rho(\theta)\right]  \\
&  =\operatorname{Tr}\left[  \left\{  (\partial_{i}G(\theta)),\Psi_{\theta
}(\partial_{j}G(\theta))\right\}  \rho(\theta)\right]  \\
&  =\left\langle \left\{  \partial_{i}G(\theta)),\Psi_{\theta}(\partial
_{j}G(\theta))\right\}  \right\rangle _{\rho(\theta)}.
\end{align}
The third equality follows because $\left[  e^{iG(\theta)t},\rho
(\theta)\right]  =0$ for all $t\in\mathbb{R}$. The fourth equality follows by
setting $\tau=t_{1}$ and $t=t_{2}-t_{1}$. The fifth equality uses the
definition of $q(t)$ in~\eqref{eq:convolve-high-peak-tent}. 

\section{Proof of Proposition~\ref{prop:main-sld}}

\label{app:proof-SLD}

Here we follow the proof of \cite[Theorem~3]{patel2024naturalgradient} quite closely, with the main change being that $G_j$ gets replaced by $\partial_j G(\theta)$. We provide a detailed proof for completeness.

Recall that the $(k, \ell)$-th element of the SLD operator $L^{(j)}(\theta)$ is given as follows \cite[Eq.~(86)]{Sidhu2020}:
\begin{align}
    L_{k\ell}^{(j)}(\theta) = \frac{2\langle k|\partial_{j}\rho(\theta
)|\ell\rangle}{\lambda
_{k}+\lambda_{\ell}}.
\end{align}
Plugging~\eqref{eq:basic-calc-deriv-HC} into the above equation, we find that
\begin{align}
    L_{k\ell}^{(j)}(\theta) & = \frac{2\left(  -\frac{1}{2}\langle k|\Phi_{\theta}(\partial_j G(\theta))|\ell\rangle\left(
\lambda_{k}+\lambda_{\ell}\right)  +\delta_{k,\ell}\lambda_{\ell}\left\langle
\partial_j G(\theta)\right\rangle _{\rho(\theta)}\right) }{\lambda
_{k}+\lambda_{\ell}}\\
& = \frac{ -\langle k|\Phi_{\theta}(\partial_j G(\theta))|\ell\rangle\left(
\lambda_{k}+\lambda_{\ell}\right)}{\lambda
_{k}+\lambda_{\ell}} + \frac{2\delta_{k,\ell}\lambda_{\ell}\left\langle
\partial_j G(\theta)\right\rangle _{\rho(\theta)}}{\lambda
_{k}+\lambda_{\ell}}\\
& = -\langle k|\Phi_{\theta}(\partial_j G(\theta))|\ell\rangle + \frac{2\lambda_{k}\left\langle
\partial_j G(\theta)\right\rangle _{\rho(\theta)}}{2\lambda
_{k}}\delta_{k,\ell}\\
& = -\langle k|\Phi_{\theta}(\partial_j G(\theta))|\ell\rangle + \left\langle
\partial_j G(\theta)\right\rangle _{\rho(\theta)}\delta_{k,\ell}.
\end{align}
This implies that
\begin{align}
    L^{(j)}(\theta) & = \sum_{k, \ell}L_{k\ell}^{(j)}(\theta) |k\rangle\!\langle\ell|\\
    &= \sum_{k, \ell}\left(-\langle k|\Phi_{\theta}(\partial_j G(\theta))|\ell\rangle + \left\langle
\partial_j G(\theta)\right\rangle _{\rho(\theta)}\delta_{k,\ell}\right) |k\rangle\!\langle\ell|\\
& = -\Phi_{\theta}(\partial_j G(\theta)) + \left\langle
\partial_j G(\theta)\right\rangle _{\rho(\theta)} I.
\end{align}

\section{Action of quadratic Hamiltonians on quadrature operators}

\begin{lemma}
\label{lem:BCH-etc}For $\hat{H}\equiv\frac{1}{2}(\hat{x}^{c})^{T}H\hat{x}^{c}$
and $\hat{x}^{c}\equiv\hat{x}-\mu$, the following equality holds:
\begin{equation}
e^{i\hat{H}t}\hat{x}_{k}^{c}e^{-i\hat{H}t}=\sum_{\ell}S_{k,\ell}(t)\hat
{x}_{\ell}^{c},
\end{equation}
where $S(t)\coloneqq e^{\Omega Ht}$.
\end{lemma}

\begin{proof}
{ The following identity is known as Hadamard's lemma and holds for operators
$X$ and $Y$:
\begin{equation}
e^{X}Ye^{-X}=\sum_{n=0}^{\infty}\frac{\left[  \left(  X\right)  ^{n},Y\right]
}{n!},
\end{equation}
where the nested commutator is defined as
\begin{equation}
\lbrack(X)^{n},Y]\coloneqq\underbrace{[X,\dotsb\lbrack X,[X}_{n\text{ times }
},Y]]\dotsb],\quad\quad\lbrack(X)^{0},Y]\coloneqq Y.\label{eq:def-nested-comm}
\end{equation}
Employing this, we find that
\begin{equation}
e^{i\hat{H}t}\hat{x}_{k}^{c}e^{-i\hat{H}t}=\sum_{n=0}^{\infty}\frac{\left[
\left(  i\hat{H}t\right)  ^{n},\hat{x}_{k}^{c}\right]  }{n!}.
\end{equation}
The zeroth-order term in the series is
\begin{equation}
\hat{x}_{k}^{c}=\left[  I\ \hat{x}^{c}\right]  _{k},
\end{equation}
where $\hat{x}^{c}\equiv\left(  \hat{x}_{1}^{c},\ldots,\hat{x}_{2n}
^{c}\right)  $ is a vector of quadrature operators, $I$ is the $2n\times2n$
identity matrix, and the notation above indicates the $k$th entry of the
vector $I\ \hat{x}^{c}=\hat{x}^{c}$. To determine the first-order term in the
series, consider that
\begin{align}
\left[  i\hat{H}t,\hat{x}_{k}^{c}\right]   &  =it\left[  \frac{1}{2}(\hat
{x}^{c})^{T}H\hat{x}^{c},\hat{x}_{k}^{c}\right]  \\
&  =\frac{it}{2}\left[  \sum_{i,j}h_{ij}\hat{x}_{i}^{c}\hat{x}_{j}^{c},\hat
{x}_{k}^{c}\right]  \\
&  =\frac{it}{2}\sum_{i,j}h_{ij}\left[  \hat{x}_{i}^{c}\hat{x}_{j}^{c},\hat
{x}_{k}^{c}\right]  .\label{eq:ham-commutators-hadamard}
\end{align}
Then
\begin{align}
\left[  \hat{x}_{i}^{c}\hat{x}_{j}^{c},\hat{x}_{k}^{c}\right]   &  =\hat
{x}_{i}^{c}\hat{x}_{j}^{c}\hat{x}_{k}^{c}-\hat{x}_{k}^{c}\hat{x}_{i}^{c}
\hat{x}_{j}^{c}\\
&  =\hat{x}_{i}^{c}\hat{x}_{j}^{c}\hat{x}_{k}^{c}-\hat{x}_{i}^{c}\hat{x}
_{k}^{c}\hat{x}_{j}^{c}+\hat{x}_{i}^{c}\hat{x}_{k}^{c}\hat{x}_{j}^{c}-\hat
{x}_{k}^{c}\hat{x}_{i}^{c}\hat{x}_{j}^{c}\\
&  =\hat{x}_{i}^{c}\left[  \hat{x}_{j}^{c},\hat{x}_{k}^{c}\right]  +\left[
\hat{x}_{i}^{c},\hat{x}_{k}^{c}\right]  \hat{x}_{j}^{c}\\
&  =i\hat{x}_{i}^{c}\Omega_{j,k}+i\hat{x}_{j}^{c}\Omega_{i,k}\\
&  =i\left(  \hat{x}_{i}^{c}\Omega_{j,k}+\hat{x}_{j}^{c}\Omega_{i,k}\right)  .
\end{align}
Plugging back into~\eqref{eq:ham-commutators-hadamard}, we find that
\begin{align}
\left[  i\hat{H}t,\hat{x}_{k}^{c}\right]   &  =\left(  it\right)  \frac{i}
{2}\sum_{i,j}h_{ij}\left(  \hat{x}_{i}^{c}\Omega_{j,k}+\hat{x}_{j}^{c}
\Omega_{i,k}\right)  \\
&  =-\frac{t}{2}\left(  \sum_{i,j}h_{ij}\hat{x}_{i}^{c}\Omega_{j,k}+\sum
_{i,j}h_{ij}\hat{x}_{j}^{c}\Omega_{i,k}\right)  \\
&  =-\frac{t}{2}\left(  \sum_{i,j}-\Omega_{k,j}h_{ji}\hat{x}_{i}^{c}
+\sum_{i,j}-\Omega_{k,i}h_{ij}\hat{x}_{j}^{c}\right)  \\
&  =t\left(  \sum_{i,j}\Omega_{k,j}h_{ji}\hat{x}_{i}^{c}\right)  \\
&  =\left[  \left(  \Omega Ht\right)  \hat{x}^{c}\right]  _{k}.
\end{align}
The third equality follows because $H=H^{T}$ and $\Omega=-\Omega^{T}$. So this
means that the first-order term in the series is
\begin{equation}
\left[  i\hat{H}t,\hat{x}_{k}^{c}\right]  =\left[  \left(  \Omega Ht\right)
\hat{x}^{c}\right]  _{k}.
\end{equation}
}

{ To determine the second-order term, consider that
\begin{align}
\lbrack i\hat{H}t,\left[  i\hat{H}t,\hat{x}_{k}^{c}\right]  ]  &  =\left[
i\hat{H}t,\left[  \left(  \Omega Ht\right)  \hat{x}^{c}\right]  _{k}\right] \\
&  =t\left[  i\hat{H}t,\left(  \sum_{i,j}\Omega_{k,j}h_{ji}\hat{x}_{i}
^{c}\right)  \right] \\
&  =t\sum_{i,j}\Omega_{k,j}h_{ji}\left[  i\hat{H}t,\hat{x}_{i}^{c}\right] \\
&  =t\sum_{i,j}\Omega_{k,j}h_{ji}t\left(  \sum_{\ell,m}\Omega_{i,\ell}h_{\ell
m}\hat{x}_{m}^{c}\right) \\
&  =t^{2}\sum_{i,j,\ell,m}\Omega_{k,j}h_{ji}\Omega_{i,\ell}h_{\ell m}\hat
{x}_{m}^{c}\\
&  =\left[  \left(  \Omega Ht\right)  ^{2}\hat{x}^{c}\right]  _{k}.
\end{align}
So the second-order term is
\begin{equation}
\frac{1}{2!}\left[  \left(  \Omega Ht\right)  ^{2}\hat{x}^{c}\right]  _{k}.
\end{equation}
}

{ This pattern continues, so that the $n$th order term is
\begin{equation}
\frac{1}{n!}\left[  \left(  \Omega Ht\right)  ^{n}\hat{x}^{c}\right]  _{k},
\end{equation}
which finally implies that
\begin{align}
e^{i\hat{H}t}\hat{x}_{k}^{c}e^{-i\hat{H}t}  &  =\sum_{n=0}^{\infty}
\frac{\left[  \left(  i\hat{H}t\right)  ^{n},\hat{x}_{k}^{c}\right]  }{n!}\\
&  =\sum_{n=0}^{\infty}\frac{\left[  \left(  \Omega Ht\right)  ^{n}\hat{x}
^{c}\right]  _{k}}{n!}\\
&  =\left[  \sum_{n=0}^{\infty}\frac{1}{n!}\left(  \Omega Ht\right)  ^{n}
\hat{x}^{c}\right]  _{k}\\
&  =\left[  e^{\Omega Ht}\hat{x}^{c}\right]  _{k}\\
&  =\sum_{\ell}S_{k,\ell}(t)\hat{x}_{\ell}^{c},
\end{align}
as claimed.}
\end{proof}

\section{Proof of Theorem~\ref{thm:main}}

\label{app:main-thm-proof}Let us define the following quantum channel:
\begin{equation}
\Psi_{\mu,H}(X)\coloneqq\int_{\mathbb{R}}dt\ q(t)\ \exp\!\left(  \frac{i}
{2}(\hat{x}^{c})^{T}H\hat{x}^{c}t\right)  X\exp\!\left(  -\frac{i}{2}(\hat
{x}^{c})^{T}H\hat{x}^{c}t\right)  ,
\end{equation}
which is defined in terms of~\eqref{eq:Psi-ch} by making the substitions in
\eqref{eq:Gauss-sub-1} and~\eqref{eq:Gauss-sub-2}. 
We conclude from \eqref{eq:derivs-Gaussian-Ham}--\eqref{eq:derivs-Gaussian-Ham-2} and
Proposition~\ref{prop:FB-info-mats-alt} that
\begin{align}
I_{\left(  k_{1},\ell_{1}\right)  ,\left(  k_{2},\ell_{2}\right)
}^{\operatorname{FB}}(\mu,H) &  =\frac{1}{32}\left\langle \left\{  \left\{
\hat{x}_{k_{1}}^{c},\hat{x}_{\ell_{1}}^{c}\right\}  ,\Psi_{\mu,H}\left(
\left\{  \hat{x}_{k_{2}}^{c},\hat{x}_{\ell_{2}}^{c}\right\}  \right)
\right\}  \right\rangle _{\rho(\mu,H)}
-\frac{1}{16}\left\langle \left\{  \hat{x}_{k_{1}}^{c},\hat{x}
_{\ell_{1}}^{c}\right\}  \right\rangle _{\rho(\mu,H)}\left\langle \left\{
\hat{x}_{k_{2}}^{c},\hat{x}_{\ell_{2}}^{c}\right\}  \right\rangle _{\rho
(\mu,H)},\\
I_{\left(  k,\ell\right)  ,m}^{\operatorname{FB}}(\mu,H) &  =\frac{1}{8}\left\langle \left\{  \left\{  \hat{x}_{k_{1}}^{c},\hat{x}_{\ell_{1}}
^{c}\right\}  ,\Psi_{\mu,H}\left(  -\sum_{j=1}^{2n}h_{m,j}\hat{x}_{j}
^{c}\right)  \right\}  \right\rangle _{\rho(\mu,H)}\nonumber\\
&  \qquad-\frac{1}{4}\left\langle \left\{  \hat{x}_{k_{1}}^{c},\hat{x}
_{\ell_{1}}^{c}\right\}  \right\rangle _{\rho(\mu,H)}\left\langle \left(
-\sum_{j=1}^{2n}h_{m,j}\hat{x}_{j}^{c}\right)  \right\rangle _{\rho(\mu,H)},\\
I_{m,\left(  k,\ell\right)  }^{\operatorname{FB}}(\mu,H) &  =I_{\left(
k,\ell\right)  ,m}(\mu,H),\\
I_{m_{1},m_{2}}^{\operatorname{FB}}(\mu,H) &  =\frac{1}{2}\left\langle \left\{
-\sum_{j_{1}=1}^{2n}h_{m_{1},j_{1}}\hat{x}_{j_{1}}^{c},\Psi_{\mu,H}\left(
-\sum_{j_{2}=1}^{2n}h_{m_{2},j_{2}}\hat{x}_{j_{2}}^{c}\right)  \right\}
\right\rangle _{\rho(\mu,H)}\nonumber\\
&  \qquad-\left\langle \left(  -\sum_{j_{1}=1}^{2n}h_{m_{1},j_{1}}\hat
{x}_{j_{1}}^{c}\right)  \right\rangle _{\rho(\mu,H)}\left\langle \left(
-\sum_{j_{2}=1}^{2n}h_{m_{2},j_{2}}\hat{x}_{j_{2}}^{c}\right)  \right\rangle
_{\rho(\mu,H)}.
\end{align}

The terms $I_{\left(  k,\ell\right)  ,m}^{\operatorname{FB}}(\mu,H)$ and
$I_{m,\left(  k,\ell\right)  }^{\operatorname{FB}}(\mu,H)$ are actually equal
to zero, because they exclusively involve third-order products of the
quadrature operators, which are known to vanish by the arguments of
(C17)--(C20) of~\cite{WTLB17}. This establishes~\eqref{eq:zero-cross-terms-nice}.

For the second term of $I_{\left(  k_{1},\ell_{1}\right)  ,\left(  k_{2}
,\ell_{2}\right)  }^{\operatorname{FB}}(\mu,H)$, we find that
\begin{align}
&  \frac{1}{16}\left\langle \left\{  \hat{x}_{k_{1}}^{c},\hat{x}_{\ell_{1}
}^{c}\right\}  \right\rangle _{\rho(\mu,H)}\left\langle \left\{  \hat
{x}_{k_{2}}^{c},\hat{x}_{\ell_{2}}^{c}\right\}  \right\rangle _{\rho(\mu
,H)}\nonumber\\
&  =\frac{1}{4}\left\langle \frac{1}{2}\left\{  \hat{x}_{k_{1}}^{c},\hat
{x}_{\ell_{1}}^{c}\right\}  \right\rangle _{\rho(\mu,H)}\left\langle \frac
{1}{2}\left\{  \hat{x}_{k_{2}}^{c},\hat{x}_{\ell_{2}}^{c}\right\}
\right\rangle _{\rho(\mu,H)}\\
&  =\frac{1}{4}V_{k_{1},\ell_{1}}V_{k_{2},\ell_{2}}.
\end{align}
After adopting the notation
\begin{equation}
\hat{H}\equiv\frac{1}{2}(\hat{x}^{c})^{T}H\hat{x}^{c},\label{eq:Ham-abbrev}
\end{equation}
consider that
\begin{align}
&  \frac{1}{32}\left\langle \left\{  \left\{  \hat{x}_{k_{1}}^{c},\hat
{x}_{\ell_{1}}^{c}\right\}  ,\Psi_{\mu,H}\left(  \left\{  \hat{x}_{k_{2}}
^{c},\hat{x}_{\ell_{2}}^{c}\right\}  \right)  \right\}  \right\rangle
_{\rho(\mu,H)}\nonumber\\
&  =\frac{1}{ 32}\int_{\mathbb{R}}dt\ q(t)\ \left\langle \left\{  \left\{
\hat{x}_{k_{1}}^{c},\hat{x}_{\ell_{1}}^{c}\right\}  ,e^{i\hat{H}t}\left\{
\hat{x}_{k_{2}}^{c},\hat{x}_{\ell_{2}}^{c}\right\}  e^{-i\hat{H}t}\right\}
\right\rangle _{\rho(\mu,H)}\\
&  =\frac{1}{ 32}\int_{\mathbb{R}}dt\ q(t)\ \left\langle \left\{  \left\{
\hat{x}_{k_{1}}^{c},\hat{x}_{\ell_{1}}^{c}\right\}  ,\left\{  e^{i\hat{H}
t}\hat{x}_{k_{2}}^{c}e^{-i\hat{H}t},e^{i\hat{H}t}\hat{x}_{\ell_{2}}
^{c}e^{-i\hat{H}t}\right\}  \right\}  \right\rangle _{\rho(\mu,H)}.
\end{align}
Now consider that, by an application of Lemma~\ref{lem:BCH-etc}, it follows
that
\begin{equation}
e^{i\hat{H}t}\hat{x}_{k}^{c}e^{-i\hat{H}t}=\sum_{\ell}S_{k,\ell}(t)\hat {x}_{\ell}^{c},\label{eq:symplectic-Gauss-Ham}
\end{equation}
where $S(t)\coloneqq e^{\Omega Ht}$. (See also Eq.~(3.22) of~\cite{Ser17} for
the uncentered case.) Then we find that
\begin{align}
&  \frac{1}{ 32}\left\langle \left\{  \left\{  \hat{x}_{k_{1}}^{c},\hat
{x}_{\ell_{1}}^{c}\right\}  ,\left\{  e^{i\hat{H}t}\hat{x}_{k_{2}}
^{c}e^{-i\hat{H}t},e^{i\hat{H}t}\hat{x}_{\ell_{2}}^{c}e^{-i\hat{H}t}\right\}
\right\}  \right\rangle _{\rho(\mu,H)}\nonumber\\
&  =\frac{1}{ 32}\left\langle \left\{  \left\{  \hat{x}_{k_{1}}^{c},\hat
{x}_{\ell_{1}}^{c}\right\}  ,\left\{  \sum_{i_{2}}S_{k_{2},i_{2}}(t)\hat
{x}_{i_{2}}^{c},\sum_{j_{2}}S_{\ell_{2},j_{2}}(t)\hat{x}_{j_{2}}^{c}\right\}
\right\}  \right\rangle _{\rho(\mu,H)}\\
&  =\sum_{i_{2},j_{2}}S_{k_{2},i_{2}}(t)S_{\ell_{2},j_{2}}(t)\frac{1}{ 32}\left\langle \left\{  \left\{  \hat{x}_{k_{1}}^{c},\hat{x}_{\ell_{1}}
^{c}\right\}  ,\left\{  \hat{x}_{i_{2}}^{c},\hat{x}_{j_{2}}^{c}\right\}
\right\}  \right\rangle _{\rho(\mu,H)}\\
&  =\sum_{i_{2},j_{2}}S_{k_{2},i_{2}}(t)S_{\ell_{2},j_{2}}(t)\frac{1}
{ 4}\left\langle \left(  \hat{x}_{k_{1}}^{c}\circ\hat{x}_{\ell_{1}}^{c}\right)
\circ\left(  \hat{x}_{i_{2}}^{c}\circ\hat{x}_{j_{2}}^{c}\right)  \right\rangle
_{\rho(\mu,H)}\\
&  =\sum_{i_{2},j_{2}}S_{k_{2},i_{2}}(t)S_{\ell_{2},j_{2}}(t)\frac{1}
{ 4}\left(
\begin{array}
[c]{c}
V_{k_{1},\ell_{1}}V_{i_{2},j_{2}}+V_{k_{1},i_{2}}V_{\ell_{1},j_{2}}
+V_{k_{1},j_{2}}V_{\ell_{1},i_{2}}\\
-\frac{1}{4}\Omega_{k_{1},i_{2}}\Omega_{\ell_{1},j_{2}}-\frac{1}{4}
\Omega_{k_{1},j_{2}}\Omega_{\ell_{1},i_{2}}
\end{array}
\right)  \\
&  =\frac{1}{ 4}\left(
\begin{array}
[c]{c}
V_{k_{1},\ell_{1}}\left[  S(t)VS(t)^{T}\right]  _{k_{2},\ell_{2}}+\left[
VS(t)^{T}\right]  _{k_{1},k_{2}}\left[  VS(t)^{T}\right]  _{\ell_{1},\ell_{2}
}\\
+\left[  VS(t)^{T}\right]  _{k_{1},\ell_{2}}\left[  VS(t)^{T}\right]
_{\ell_{1},k_{2}}-\frac{1}{4}\left[  \Omega S(t)^{T}\right]  _{k_{1},k_{2}
}\left[  \Omega S(t)^{T}\right]  _{\ell_{1},\ell_{2}}\\
-\frac{1}{4}\left[  \Omega S(t)^{T}\right]  _{k_{1},\ell_{2}}\left[  \Omega
S(t)^{T}\right]  _{\ell_{1},k_{2}}
\end{array}
\right)  .
\end{align}
In the above, we made use of Eq.~(E11) of~\cite{MI10}. Thus we conclude that
\begin{multline}
I_{\left(  k_{1},\ell_{1}\right)  ,\left(  k_{2},\ell_{2}\right)
}^{\operatorname{FB}}(\mu,H)=\label{eq:cross-term-Ham-mat-pf}\\
\frac{1}{ 4}\int_{\mathbb{R}}dt\ q(t)\left(
\begin{array}
[c]{c}
V_{k_{1},\ell_{1}}\left[  S(t)VS(t)^{T}\right]  _{k_{2},\ell_{2}}+\left[
VS(t)^{T}\right]  _{k_{1},k_{2}}\left[  VS(t)^{T}\right]  _{\ell_{1},\ell_{2}
}\\
+\left[  VS(t)^{T}\right]  _{k_{1},\ell_{2}}\left[  VS(t)^{T}\right]
_{\ell_{1},k_{2}}-\frac{1}{4}\left[  \Omega S(t)^{T}\right]  _{k_{1},k_{2}
}\left[  \Omega S(t)^{T}\right]  _{\ell_{1},\ell_{2}}\\
-\frac{1}{4}\left[  \Omega S(t)^{T}\right]  _{k_{1},\ell_{2}}\left[  \Omega
S(t)^{T}\right]  _{\ell_{1},k_{2}}
\end{array}
\right)  -\frac{1}{4}V_{k_{1},\ell_{1}}V_{k_{2},\ell_{2}},
\end{multline}
which establishes~\eqref{eq:cross-term-Ham-mat}.

Finally,
\begin{multline}
I_{m_{1},m_{2}}^{\operatorname{FB}}(\mu,H)=\frac{1}{2} \left\langle \left\{  -\sum
_{j_{1}=1}^{2n}h_{m_{1},j_{1}}\hat{x}_{j_{1}}^{c},\Psi_{\mu,H}\left(
-\sum_{j_{2}=1}^{2n}h_{m_{2},j_{2}}\hat{x}_{j_{2}}^{c}\right)  \right\}
\right\rangle _{\rho(\mu,H)}\\
-\left\langle \left(  -\sum_{j_{1}=1}^{2n}h_{m_{1},j_{1}}\hat{x}_{j_{1}}
^{c}\right)  \right\rangle _{\rho(\mu,H)}\left\langle \left(  -\sum_{j_{2}
=1}^{2n}h_{m_{2},j_{2}}\hat{x}_{j_{2}}^{c}\right)  \right\rangle _{\rho
(\mu,H)}.
\end{multline}
So then
\begin{align}
&  \left\langle \left\{  -\sum_{j_{1}=1}^{2n}h_{m_{1},j_{1}}\hat{x}_{j_{1}
}^{c},\Psi_{\mu,H}\left(  -\sum_{j_{2}=1}^{2n}h_{m_{2},j_{2}}\hat{x}_{j_{2}
}^{c}\right)  \right\}  \right\rangle _{\rho(\mu,H)}\nonumber\\
&  =\sum_{j_{1},j_{2}=1}^{2n}h_{m_{1},j_{1}}h_{m_{2},j_{2}}\left\langle
\left\{  \hat{x}_{j_{1}}^{c},\Psi_{\mu,H}\left(  \hat{x}_{j_{2}}^{c}\right)
\right\}  \right\rangle _{\rho(\mu,H)}\\
&  =\sum_{j_{1},j_{2}=1}^{2n}h_{m_{1},j_{1}}h_{m_{2},j_{2}}\int_{\mathbb{R}
}dt\ q(t)\left\langle \left\{  \hat{x}_{j_{1}}^{c},e^{i\hat{H}t}\hat{x}
_{j_{2}}^{c}e^{-i\hat{H}t}\right\}  \right\rangle _{\rho(\mu,H)}\\
&  =\sum_{j_{1},j_{2}=1}^{2n}h_{m_{1},j_{1}}h_{m_{2},j_{2}}\int_{\mathbb{R}
}dt\ q(t)\left\langle \left\{  \hat{x}_{j_{1}}^{c},\sum_{\ell}S_{j_{2},\ell
}(t)\hat{x}_{\ell}^{c}\right\}  \right\rangle _{\rho(\mu,H)}\\
&  =\sum_{j_{1},j_{2}=1}^{2n}h_{m_{1},j_{1}}h_{m_{2},j_{2}}\int_{\mathbb{R}
}dt\ q(t)\sum_{\ell}S_{j_{2},\ell}(t)\left\langle \left\{  \hat{x}_{j_{1}}
^{c},\hat{x}_{\ell}^{c}\right\}  \right\rangle _{\rho(\mu,H)}\\
&  =\sum_{j_{1},j_{2}=1}^{2n}h_{m_{1},j_{1}}h_{m_{2},j_{2}}\int_{\mathbb{R}
}dt\ q(t)\sum_{\ell}S_{j_{2},\ell}(t)2V_{j_{1},\ell}\\
&  =2\sum_{j_{1},j_{2}=1}^{2n}h_{m_{1},j_{1}}h_{m_{2},j_{2}}\int_{\mathbb{R}
}dt\ q(t)\left[  VS(t)^{T}\right]  _{j_{1},j_{2}}\\
&  =2\int_{\mathbb{R}}dt\ q(t)\left[  HVS(t)^{T}H\right]  _{m_{1},m_{2}}.
\end{align}
The third equality follows from Lemma~\ref{lem:BCH-etc}, with $S(t)=e^{\Omega
Ht}$. Additionally,
\begin{align}
&  \left\langle \left(  -\sum_{j_{1}=1}^{2n}h_{m_{1},j_{1}}\hat{x}_{j_{1}}
^{c}\right)  \right\rangle _{\rho(\mu,H)}\left\langle \left(  -\sum_{j_{2}
=1}^{2n}h_{m_{2},j_{2}}\hat{x}_{j_{2}}^{c}\right)  \right\rangle _{\rho
(\mu,H)}\nonumber\\
&  =\sum_{j_{1}=1}^{2n}h_{m_{1},j_{1}}\sum_{j_{2}=1}^{2n}h_{m_{2},j_{2}
}\left\langle \hat{x}_{j_{1}}^{c}\right\rangle _{\rho(\mu,H)}\left\langle
\hat{x}_{j_{2}}^{c}\right\rangle _{\rho(\mu,H)}\\
&  =0.
\end{align}
So we finally have that
\begin{equation}
I_{m_{1},m_{2}}^{\operatorname{FB}}(\mu,H)=\int_{\mathbb{R}}dt\ q(t)\left[
HVS(t)^{T}H\right]  _{m_{1},m_{2}}.\label{eq:cross-term-mean-pf}
\end{equation}
Finally, by noticing that
\begin{equation}
S(t)^{T}=\left(  e^{\Omega Ht}\right)  ^{T}=e^{\left(  \Omega H\right)  ^{T}
t}=e^{-H\Omega t},
\end{equation}
and redefining $S(t)$ to be $S(t)^{T}$, we arrive
at~\eqref{eq:cross-term-Ham-mat} and~\eqref{eq:cross-term-mean} from
\eqref{eq:cross-term-Ham-mat-pf}\ and~\eqref{eq:cross-term-mean-pf},
respectively. This concludes the proof of the Fisher--Bures information matrix elements.

The proof for the Kubo--Mori information matrix elements indeed follows by the
substitution $q(t)\rightarrow p(t)$, due to this being the only difference
between the formulas in~\eqref{eq:KM-info-mat-gen-thermal} and~\eqref{eq:FB-info-mat-gen-thermal-alt}.

\section{Proof of Theorem~\ref{thm:derivative-bosonic-Gaussian}}

\label{app:deriv-bos-G-therm}Define the channel
\begin{equation}
\Phi_{\mu,H}(X)\coloneqq\int_{\mathbb{R}}dt\ p(t)\ e^{i\hat{H}t}Xe^{-i\hat
{H}t},
\end{equation}
where $\hat{H}$ is defined in~\eqref{eq:Ham-abbrev} and $p(t)$ in
\eqref{eq:high-peak-tent}. By applying Proposition~\ref{prop:deriv} and
\eqref{eq:derivs-Gaussian-Ham}, we find that
\begin{align}
\frac{\partial}{\partial h_{k,\ell}}\rho(\mu,H)  & =-\frac{1}{2}\left\{
\Phi_{\mu,H}\left(  \frac{1}{4}\left\{  \hat{x}_{k}^{c},\hat{x}_{\ell}
^{c}\right\}  \right)  ,\rho(\mu,H)\right\}  +\rho(\mu,H)\left\langle \frac
{1}{4}\left\{  \hat{x}_{k}^{c},\hat{x}_{\ell}^{c}\right\}  \right\rangle
_{\rho(\mu,H)}\\
& =-\frac{1}{8}\int_{\mathbb{R}}dt\ p(t)\left\{  e^{i\hat{H}t}\left\{  \hat
{x}_{k}^{c},\hat{x}_{\ell}^{c}\right\}  e^{-i\hat{H}t},\rho(\mu,H)\right\}
+\frac{1}{2}\rho(\mu,H)\ V_{k,\ell}\\
& =-\frac{1}{8}\int_{\mathbb{R}}dt\ p(t)\left\{  \left\{  e^{i\hat{H}t}\hat
{x}_{k}^{c}e^{-i\hat{H}t},e^{i\hat{H}t}\hat{x}_{\ell}^{c}e^{-i\hat{H}
t}\right\}  ,\rho(\mu,H)\right\}  +\frac{1}{2}\rho(\mu,H)\ V_{k,\ell}\\
& =-\frac{1}{8}\int_{\mathbb{R}}dt\ p(t)\left\{  \left\{  \sum_{k^{\prime}
}S_{k,k^{\prime}}(t)\hat{x}_{k^{\prime}}^{c},\sum_{\ell^{\prime}}S_{\ell
,\ell^{\prime}}(t)\hat{x}_{\ell^{\prime}}^{c}\right\}  ,\rho(\mu,H)\right\}
+\frac{1}{2}\rho(\mu,H)\ V_{k,\ell}\\
& =-\frac{1}{8}\int_{\mathbb{R}}dt\ p(t)\sum_{k^{\prime},\ell^{\prime}
}S_{k,k^{\prime}}(t)S_{\ell,\ell^{\prime}}(t)\left\{  \left\{  \hat
{x}_{k^{\prime}}^{c},\hat{x}_{\ell^{\prime}}^{c}\right\}  ,\rho(\mu
,H)\right\}  +\frac{1}{2}\rho(\mu,H)\ V_{k,\ell}.
\end{align}
This establishes~\eqref{eq:deriv-wrt-Ham}.

Now applying Proposition~\ref{prop:deriv} and
\eqref{eq:derivs-Gaussian-Ham-2}, we find that
\begin{align}
\frac{\partial}{\partial\mu_{m}}\rho(\mu,H)  & =-\frac{1}{2}\left\{  \Phi
_{\mu,H}\left(  -\sum_{j=1}^{2n}h_{m,j}\hat{x}_{j}^{c}\right)  ,\rho
(\mu,H)\right\}  +\rho(\mu,H)\left\langle -\sum_{j=1}^{2n}h_{m,j}\hat{x}
_{j}^{c}\right\rangle _{\rho(\mu,H)}\\
& =-\frac{1}{2}\left\{  \Phi_{\mu,H}\left(  -\sum_{j=1}^{2n}h_{m,j}\hat{x}
_{j}^{c}\right)  ,\rho(\mu,H)\right\}  \\
& =\frac{1}{2}\sum_{j=1}^{2n}h_{m,j}\int_{\mathbb{R}}dt\ p(t)\left\{
e^{i\hat{H}t}\hat{x}_{j}^{c}e^{-i\hat{H}t},\rho(\mu,H)\right\}  \\
& =\frac{1}{2}\int_{\mathbb{R}}dt\ p(t)\sum_{j,j^{\prime}}h_{m,j}
S_{j,j^{\prime}}(t)\left\{  \hat{x}_{j^{\prime}}^{c},\rho(\mu,H)\right\}  ,
\end{align}
thus concluding the proof of~\eqref{eq:deriv-wrt-mean}.

\section{Fidelity results comparison}

 As a consistency check, we compare our formulation to that of Ref.~\cite{paraoanu2000fidelity}; the fidelity between two displaced squeezed thermal states is given by
\begin{align}
F(\rho_1,\rho_2) &= \frac{2}{\sqrt{\Delta + T} - \sqrt{T}} \exp\left[ -\frac{1}{2}u^T(V_1 + V_2)^{-1} u \right], \label{eq:scutaru-fid} \\
          \Delta &= \det ( 2(V_1 + V_2)), \qquad T = (\det (2V_1) - 1)(\det (2V_2) - 1).
\end{align}

Consider the zero-displacement one-mode squeezed thermal state with fixed mean photon number $\epsilon>0$ and covariance matrix. Fidelity is invariant under unitary operations; therefore we can set the rotation angle to zero for simplicity.
Then
\begin{align}
V(r)=\nu
\begin{pmatrix}
e^{2r} & 0\\
0 & e^{-2r}
\end{pmatrix}, \qquad \nu = (1+\epsilon)/2.
\end{align}

We compare two nearby states with squeezing parameters $r$ and $r+dr$, so that
\begin{align}
V_1=V(r),\qquad V_2=V(r+dr).
\end{align}
Since the relative displacement is zero, the fidelity formula in \eqref{eq:scutaru-fid} reduces to
\begin{align}
F(\rho_r,\rho_{r+dr})
=
\frac{2}{\sqrt{\Delta+T}-\sqrt{T}},
\end{align}

For this family, $\det V(r)=\nu^2$, so that
\begin{align}
T=(4\nu^2-1)^2.
\end{align}
Furthermore,
\begin{align}
\det(V_1+V_2)
&=
\nu^2\bigl(e^{2r}+e^{2r+2dr}\bigr)\bigl(e^{-2r}+e^{-2r-2dr}\bigr) \\
&=
4\nu^2\cosh^2(dr),
\end{align}
and therefore
\begin{align}
\Delta=16\nu^2\cosh^2(dr).
\end{align}
Thus,
\begin{align}
F(\rho_r,\rho_{r+dr})
=
\frac{2}{
\sqrt{\,16\nu^2\cosh^2(dr)+(4\nu^2-1)^2\,}
-(4\nu^2-1)
}.
\end{align}

First, using the Fisher--Bures convention adopted in this paper,
\begin{align}
d_B^2=2(1-\sqrt{F})
=
\frac14 I^{\mathrm{FB}}_{rr}\,dr^2,
\end{align}
expanding for $dr$,
\begin{align}
d^2_B = 2(1-\sqrt{F(\rho_r,\rho_{r+dr}})=1- \frac{4  \nu ^2}{4 \nu ^2+1} dr^2+O(dr^4)
\end{align}
we obtain
\begin{align}
I_{rr}
=
\frac{16\nu^2}{4\nu^2+1}.
\end{align}
If we write $\nu=(1+\epsilon)/2$, as in Sec.~VII, then this becomes
\begin{align}
\label{eq:I_rr}
I_{rr}
=
\frac{4(1+\epsilon)^2}{(1+\epsilon)^2+1}.
\end{align}

Since we have a single parameter $r$ that the Hamiltonian entries depend on: $H = H(r)$
\begin{align}
H = \log\left( \frac{\epsilon+2}{\epsilon}\right) 
\left[
\begin{array}{cc}
    e^{-2 r}  & 0 \\
   0 &     e^{2 r}  \\
\end{array}
\right];
\end{align}
therefore,
\begin{align}
h_{xx} =  \log\left( \frac{\epsilon+2}{\epsilon}\right)  e^{-2r},\qquad 
h_{yy} =-  \log\left( \frac{\epsilon+2}{\epsilon}\right)  e^{2r}, \qquad
h_{xy} = 0.
\end{align}
and the quantum Fisher information is the pullback of the Fisher information matrix in the $H$ coordinates
\begin{align}
I_{rr} = \sum_{i,j} \frac{\partial h_i}{\partial r} I_{ij}^{(H)} \frac{\partial h_j}{\partial r}.
\end{align}

To convert between the QFI matrices given in coordinates $H = (h_{xx},h_{xy},h_{yy})$ to $(r,\epsilon,\theta)$, consider that
\begin{align}
I^{(r,\epsilon,\theta)} = J^T I^{(H)} J,
\end{align}
where
\begin{align}
J_{(h_{xx},h_{xy},h_{yy}) \leftarrow (r,\epsilon,\theta)}
=
\begin{pmatrix}
\dfrac{\partial h_{xx}}{\partial r} &
\dfrac{\partial h_{xx}}{\partial \epsilon} &
\dfrac{\partial h_{xx}}{\partial \theta}
\\[10pt]
\dfrac{\partial h_{xy}}{\partial r} &
\dfrac{\partial h_{xy}}{\partial \epsilon} &
\dfrac{\partial h_{xy}}{\partial \theta}
\\[10pt]
\dfrac{\partial h_{yy}}{\partial r} &
\dfrac{\partial h_{yy}}{\partial \epsilon} &
\dfrac{\partial h_{yy}}{\partial \theta}
\end{pmatrix}.
\end{align}
Focusing on only the element $I_{rr}$,
\begin{align}
I_{rr}
=
4h_{xx}^2\, I_{(xx)(xx)}
-8h_{xx}h_{yy}\, I_{(xx)(yy)}
+4h_{yy}^2\, I_{(yy)(yy)}.
\end{align}
We compared our expression to that given by Ref.~\cite{paraoanu2000fidelity} as a function of $\epsilon$, and they are numerically equal.

\end{document}